\documentclass[runningheads]{llncs}
\usepackage{etoolbox}

\newtoggle{TR}
\toggletrue{TR}

\newtoggle{showtodos}
\toggletrue{showtodos}
\usepackage{todonotes}

\makeatletter
\RequirePackage[bookmarks,unicode,colorlinks=true]{hyperref}%
\def\@citecolor{blue}%
\def\@urlcolor{blue}%
\def\@linkcolor{blue}%

\def\orcidID#1{\smash{\href{http://orcid.org/#1}{\protect\raisebox{-1.25pt}{\protect\includegraphics{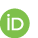}}}}}
\makeatother
\usepackage{graphicx}
\usepackage[T1]{fontenc}
\usepackage[utf8x]{inputenx}
\usepackage[american]{babel}
\usepackage{tikz}

\usepackage{amsmath}
\usepackage{thmtools, thm-restate}
\usepackage{xspace}
\usepackage{array}
\usepackage{enumitem}
\usepackage{upgreek}
\usepackage{multirow}
\usepackage{tabularx}
\usepackage{booktabs}
\usepackage{amssymb}
\usepackage{environ}
\usepackage{caption}
\usepackage{subcaption}
\usepackage{stmaryrd}
\usepackage{adjustbox}
\usepackage{makecell}
\usepackage{pgfplots}
\usepackage{pgfplotstable}

\usetikzlibrary{decorations,arrows,shapes,automata,calc}
\tikzstyle{dist}=[circle,  inner sep=1pt, fill]
\tikzstyle{point}=[circle,  inner sep=2pt, fill]
\tikzstyle{ms}=[draw,circle,text centered,minimum size=6mm,text width=3mm]
\tikzstyle{ps}=[draw,rectangle,text centered,minimum size=6mm,text width=3mm]

\usepackage[vlined,linesnumbered]{algorithm2e}
\SetAlCapSkip{2ex plus 0.5ex minus 1ex}
\SetEndCharOfAlgoLine{}
\SetArgSty{textrm}
\SetKwInOut{Input}{Input}\SetKwInOut{Output}{Output}

\SetCommentSty{commentstyle}
\SetKw{Break}{break}
\SetKwProg{Type}{type}{}{end}
\SetKwProg{Function}{function}{}{end}
\SetKwFor{ForEach}{foreach}{do}{}
\SetKwFor{WhileTrue}{repeat}{}{end}
\SetKw{Continue}{continue}
\SetKwRepeat{DoWhile}{repeat}{while}
\SetKwRepeat{DoUntil}{repeat}{until}

\newcolumntype{?}{!{\vrule width 1.5pt}}

\definecolor{color1}{RGB}{55,126,184} 
\definecolor{color2}{RGB}{228,26,28} 
\definecolor{color3}{RGB}{77,175,74} 
\definecolor{color4}{RGB}{152,78,163} 
\definecolor{color5}{RGB}{255,127,0} 
\definecolor{color6}{rgb}{0.5, 1.0, 0.83} 
\definecolor{color7}{rgb}{1.0, 0.0, 1.0} 
\definecolor{color8}{rgb}{0.66, 0.66, 0.66} 
\definecolor{color8}{rgb}{0.66, 0.66, 0.66} 
\definecolor{commentcolor}{RGB}{60,114,26}

\makeatletter
\newcolumntype{\expand}{}
\long\@namedef{NC@rewrite@\string\expand}{\expandafter\NC@find}

\NewEnviron{nproblem}[2][]{%

	\def\problem@arg{#1}%
	\def\problem@framed{framed}%
	\def\problem@lined{lined}%
	\def\problem@doublelined{doublelined}%
	\ifx\problem@arg\@empty%
	\def\problem@hline{}%
	\else%
	\ifx\problem@arg\problem@doublelined%
	\def\problem@hline{\hline\hline}%
	\else%
	\def\problem@hline{\hline}%
	\fi%
	\fi%
	\ifx\problem@arg\problem@framed%
	\def\problem@tablelayout{|>{\bfseries}lX|c}%
	\def\problem@title{\multicolumn{2}{|l|}{%
			\raisebox{-\fboxsep}{\textsc{#2}}%
	}}%
	\else
	\def\problem@tablelayout{>{\bfseries}lXc}%
	\def\problem@title{\multicolumn{2}{l}{%
			\raisebox{-\fboxsep}{\textsc{#2}}%
	}}%
	\fi%
	\goodbreak 
	\medskip\par\noindent%
	\begin{tabularx}{\textwidth}{\expand\problem@tablelayout}%
		\problem@hline%
		\problem@title\\[2\fboxsep]%
		\BODY\\\problem@hline%
	\end{tabularx}%
	\medskip\par
	\goodbreak 
}
\makeatother

\declaretheorem[name=Assumption]{assumption}

\usepackage{cleveref}

\Crefname{figure}{Fig.}{Figs.}
\crefname{figure}{fig.}{figs.}
\Crefname{tabular}{Tab.}{Tabs.}
\crefname{tabular}{tab.}{tabs.}
\Crefname{section}{Sect.}{Sects.}
\crefname{section}{sect.}{sects.}
\crefname{restr}{restriction}{restrictions}
\Crefname{restr}{Restriction}{Restrictions}
\crefname{assumption}{assumption}{assumptions}
\Crefname{assumption}{Assumption}{Assumptions}
\crefname{algorithm}{algorithm}{algorithms}
\Crefname{algorithm}{Algorithm}{Algorithms}

\setitemize{noitemsep,topsep=0pt,parsep=0pt,partopsep=0pt,leftmargin=11.0pt}
\setenumerate{noitemsep,topsep=0pt,parsep=0pt,partopsep=0pt,leftmargin=12.5pt}

\newcommand{\tool}[1]{\textsc{#1}\xspace}
\newcommand{\storm}{\tool{Storm}}
\newcommand{\multigain}{\tool{MultiGain}}
\newcommand{\prism}{\tool{PRISM}}
\newcommand{\prismgames}{\tool{PRISM-games}}

\newcommand{\benchmark}[1]{\textsf{#1}\xspace}
\newcommand{\eg}{e.g.,\xspace}
\newcommand{\ie}{i.e.,\xspace}

\newcommand{\tuple}[1]{\ensuremath{\left\langle #1 \right\rangle}}
\newcommand{\tupleaccess}[2]{{\ensuremath{#1\llbracket#2\rrbracket}}}
\newcommand{\set}[1]{\ensuremath{\left\{ #1 \right\}}}

\newcommand{\rr}{\ensuremath{\mathbb{R}}}
\newcommand{\rrnn}{\ensuremath{\mathbb{R}_{\ge0}}}
\newcommand{\rrext}{\ensuremath{\bar{\mathbb{R}}}}
\newcommand{\rrgz}{\ensuremath{\mathbb{R}_{>0}}}
\newcommand{\nn}{\ensuremath{\mathbb{N}}}

\newcommand{\ex}[1]{\ensuremath{\exists\,#1\colon\,}}
\newcommand{\fa}[1]{\ensuremath{\forall\,#1\colon\,}}
\newcommand{\dist}{\ensuremath{\mu}}
\newcommand{\dists}[1]{\ensuremath{\mathit{Dist(#1)}}}
\newcommand{\supp}[1]{\ensuremath{\mathit{supp(#1)}}}

\newcommand{\diff}{\mathop{}\!\mathrm{d}}
\newcommand\functionrestr[2]{{
  \left.\kern-\nulldelimiterspace 
  #1 
  \vphantom{\big|} 
  \right|_{#2} 
  }}

\newcommand{\ma}{\ensuremath{\mathcal{M}}}
\newcommand{\mdp}{\ensuremath{\mathcal{M}}}
\newcommand{\states}{\ensuremath{S}}
\newcommand{\actions}{\ensuremath{\mathit{Act}}}
\newcommand{\transitions}{\ensuremath{\Delta}}
\newcommand{\probabilities}{\ensuremath{\mathbf{P}}}
\newcommand{\matuple}{\ensuremath{ \tuple{\states, \actions, \transitions, \probabilities} }}
\newcommand{\sapair}{\ensuremath{\tuple{\state,\action} }}
\newcommand{\sa}[1][]{\ensuremath{ \mathit{SA}^{#1} }}
\newcommand{\ms}[1][]{\ensuremath{ \mathit{MS}^{#1} }}
\newcommand{\ps}[1][]{\ensuremath{ \mathit{PS}^{#1} }}

\newcommand{\state}{\ensuremath{s}}
\newcommand{\mstate}{\ensuremath{s}}
\newcommand{\pstate}{\ensuremath{\hat{s}}}
\newcommand{\psapair}{\ensuremath{\tuple{\pstate,\action} }}
\newcommand{\action}{\ensuremath{\alpha}}
\newcommand{\actbot}{\ensuremath{\tau}}
\newcommand{\actionordur}{\ensuremath{\kappa}}
\newcommand{\dur}{\ensuremath{t}}

\newcommand{\component}{\ensuremath{C}}
\newcommand{\mecs}[1]{\ensuremath{\mathit{MECS}(#1)}}
\newcommand{\zeromecs}[2]{\ensuremath{\mathit{MECS}_0(#1, \tuple{#2})}}
\newcommand{\componentelem}{\ensuremath{c}}
\newcommand{\componentset}{\ensuremath{\mathcal{C}}}
\newcommand{\statesof}[1]{\ensuremath{\mathit{states}(#1)}}
\newcommand{\exitsof}[1]{\ensuremath{\mathit{exits}(#1)}}
\newcommand{\submodel}[2]{\ensuremath{\tupleaccess{#1}{#2}}}

\newcommand{\quotient}[2]{\ensuremath{#1_{\setminus #2}}}
\newcommand{\staysign}{\ensuremath{\bot}}

\newcommand{\sinit}{\ensuremath{s_{\mathit{I}}}}

\newcommand{\infpath}{\ensuremath{\pi}}
\newcommand{\finpath}{\ensuremath{\hat{\pi}}}
\newcommand{\finorinfpath}{\ensuremath{\bar{\pi}}}
\newcommand{\last}[1]{\ensuremath{\mathit{last}(#1)}}
\newcommand{\durof}[1]{\ensuremath{\mathit{dur}(#1)}}
\newcommand{\infpaths}[1]{\ensuremath{\mathit{Paths}_\mathrm{inf}^{#1}}}
\newcommand{\finpaths}[1]{\ensuremath{\mathit{Paths}_\mathrm{fin}^{#1}}}

\newcommand{\lengthofpath}[1]{\ensuremath{|#1|}}
\newcommand{\prefixtime}[2]{\ensuremath{\mathit{prefix}_\mathit{time}(#1,#2)}}
\newcommand{\prefixsteps}[2]{\ensuremath{\mathit{prefix}_\mathit{steps}(#1,#2)}}

\newcommand{\strat}{\ensuremath{\sigma}}
\newcommand{\strats}[1]{\ensuremath{\Sigma^{#1}}}
\newcommand{\stratsmd}[1]{\ensuremath{\Sigma_\mathrm{md}^{#1}}}

\newcommand{\probmeasure}[2]{\ensuremath{\mathrm{Pr}^{#1}_{#2}}}
\newcommand{\expval}[2]{\ensuremath{\mathrm{Ex}^{#1}_{#2}}}

\newcommand{\rewardassignment}{\ensuremath{\mathcal{R}}}

\newcommand{\rewofpath}[2]{\ensuremath{#1(#2)}}

\newcommand{\valuefunc}{\ensuremath{f}}
\newcommand{\lrafunc}[1]{\ensuremath{\mathit{lra}(#1)}}
\newcommand{\lrafuncsteps}[1]{\ensuremath{\mathit{lra}_{\mathit{steps}}(#1)}}
\newcommand{\totfunc}[1]{\ensuremath{\mathit{tot}(#1)}}

\newcommand{\stepbound}{\ensuremath{k}}

\renewcommand{\vec}{\mathbf}
\newcommand{\pointi}{\ensuremath{p}}
\newcommand{\point}{\ensuremath{\vec{p}}}
\newcommand{\points}{\ensuremath{P}}
\newcommand{\pointidomain}{\ensuremath{\rr}}
\newcommand{\pointdomain}{\ensuremath{\pointidomain^\numobj}}
\newcommand{\cl}[1]{\ensuremath{\mathit{cl}(#1)}}
\newcommand{\downset}[1]{\ensuremath{\mathit{dwconv}(#1)}}
\newcommand{\conv}[1]{\ensuremath{\mathit{conv}(#1)}}
\newcommand{\weightvector}{\ensuremath{\vec{w}}}
\newcommand{\underapprox}{\ensuremath{P}}
\newcommand{\overapprox}{\ensuremath{Q}}

\newcommand{\numobj}{\ensuremath{\ell}}
\newcommand{\numtotobj}{{\ensuremath{k}}}
\newcommand{\objindex}{\ensuremath{j}}

\newcommand{\multiobjquery}{\ensuremath{\mathcal{F}}}
\newcommand{\lramultiobjquery}{\ensuremath{\mathcal{F}_\mathit{lra}}}
\newcommand{\ach}[2]{\ensuremath{\mathit{Ach}^{#1}(#2)}}
\newcommand{\pareto}[2]{\ensuremath{\mathit{Pareto}^{#1}(#2)}}

\hypersetup{
  pdftitle = {Multi-objective Optimization of Long-run Average and Total Rewards}
}
\iftoggle{TR}{}{
\usepackage[firstpage]{draftwatermark}
\SetWatermarkText{\hspace*{4.5in}\raisebox{6.3in}{\includegraphics[scale=0.1]{ae-badge-tacas.pdf}}}
\SetWatermarkAngle{0}
}
\usepackage{marvosym}

\begin{document}
\title{%
Multi-objective Optimization of Long-run Average and Total Rewards
}
\author{
Tim Quatmann\inst{1}$^{(\text{\Letter})}$\orcidID{0000-0002-2843-5511}
\and Joost-Pieter Katoen\inst{1}\orcidID{0000-0002-6143-1926}
}
\authorrunning{Quatmann, Katoen}
\institute{RWTH Aachen University, Aachen, Germany\\\email{tim.quatmann@cs.rwth-aachen.de}}
\date{\today}
\maketitle

\begin{abstract}
This paper presents an efficient procedure for multi-objective model checking of long-run average reward (aka: mean pay-off) and total reward objectives as well as their combination.
We consider this for Markov automata, a compositional model that captures both traditional Markov decision processes (MDPs) as well as a continuous-time variant thereof. 
The crux of our procedure is a generalization of Forejt \emph{et al.}'s approach for total rewards on MDPs to arbitrary combinations of long-run and total reward objectives on Markov automata.
Experiments with a prototypical implementation on top of the \storm model checker show encouraging results for both model types and indicate a substantial improved performance over existing multi-objective long-run MDP model checking based on linear programming.
\end{abstract}

\section{Introduction}
\paragraph{MDP model checking}
In various applications, multiple decision criteria and uncertainty frequently co-occur.
Stochastic decision processes for which the objective is to achieve multiple---possibly partly conflicting---objectives occur in various fields.
These include operations research, economics, planning in AI, and game theory, to mention a few.
This has stimulated model checking of Markov decision processes (MDPs)~\cite{Put94}, a prominent model in decision making under uncertainty, against multiple objectives.
This development enlarges the rich plethora of automated MDP verification algorithms against single objectives~\cite{BHK19}. 

\paragraph{Multi-objective MDP}
Various types of objectives known from conventional---single-objective---model checking  have been lifted to the multi-objective case.
These objectives range over $\omega$-regular specifications including LTL~\cite{EKVY07,FKNPQ11}, expected (discounted and non-discounted) total rewards~\cite{CMH06,FKNPQ11,FKP12,RSSOWW14,DKQR20}, step-bounded and reward-bounded reachability probabilities~\cite{FKP12,HJKQ20}, and---most relevant for this work---\emph{expected long-run average (LRA) rewards}~\cite{Cha07,BBCFK14,CKK17}, also known as mean pay-offs.
For the latter, all current approaches build upon linear programming (LP) which yields a theoretical time-complexity polynomial in the model size.
However, in practice, LP-based methods are often outperformed by approaches based on value- or strategy iteration~\cite{FKP12,ACDKM17,KM17}.
The LP-based approach of~\cite{FKNPQ11} and the iterative approach of~\cite{FKP12} are both implemented in \prism~\cite{KNP11} and \storm~\cite{HJKQV20}.
The LP formulation of~\cite{BBCFK14,CKK17} is implemented in \multigain~\cite{BCFK15}, an extension of \prism for multi-objective LRA rewards.

\paragraph{Contributions of this paper}
We present a computationally efficient procedure for multi-objective model checking of LRA reward and total reward objectives as well as their mixture.
The crux of our procedure is \emph{a generalization} of Forejt \emph{et al.}'s iterative approach~\cite{FKP12} for total rewards on MDPs \emph{to expected LRA reward objectives}.
In fact, our approach supports the arbitrary \emph{mixtures} of expected LRA and total reward objectives.
To our knowledge, such mixtures have not been considered so far.
Experiments on various benchmarks using a prototypical implementation in \storm indicate that this generalized iterative algorithm outperforms the LP approach implemented in \multigain. 

In addition, we extend this approach towards \emph{Markov automata} (MA)~\cite{EHZ10,DH13}, a continuous-time variant of MDP that is amenable to compositional modeling.
This model is well-suited, among others, to provide a formal semantics for dynamic fault trees and generalized stochastic Petri nets~\cite{EHKZ013}. 
Our multi-objective LRA approach for MA builds upon the value-iteration approach for single-objective expected LRA rewards on MA~\cite{BWH17} which---on practical models---outperforms the LP-based approach of~\cite{GTHRS14}.
To the best of our knowledge, this is the \emph{first multi-objective expected LRA reward approach for MA}.
Experimental results on MA benchmarks show that the treatment of a continuous-time variant of LRA comes at almost no time penalty compared to the MDP setting.


\paragraph{Other related work}
Mixtures of various other objectives have been considered for MDPs.
This includes conditional expectations or ratios of reward functions~\cite{BDKDKMW14,BDK14}.
\cite{GZ18} considers LTL formulae with probability thresholds while maximizing an expected LRA reward.
\cite{HJKQ20,KBCDDKMM18} address multi-objective quantiles on reachability properties
while \cite{RRS17,CKK17} consider multi-objective combinations of percentile queries on MDP and LRA objectives.
\cite{BDKKR17} treats resilient systems ensuring constraints on the repair mechanism while maximizing the expected LRA reward when being operational. 
The trade-off between expected LRA rewards and their variance is analyzed in~\cite{BCFK17}. \cite{HHHLT19} studies multiple objectives on interval MDP, where transition probabilities can be specified as intervals in cases where the concrete probabilities are unknown.
Multiple LRA reward objectives for \emph{stochastic games} have been treated using LP~\cite{CD16} and value iteration over convex sets~\cite{BKTW15,BKW18};
the latter is included in~\prismgames\cite{KPW18,KNPS20}.
These approaches can also be applied to MDPs when viewed as one-player stochastic games.
Algorithms for single-objective model checking of MA deal with objectives such as expected total rewards, time-bounded reachability probabilities, and expected long-run average rewards~\cite{HH12,GHHKT14,GTHRS14,BF19}. 
The only multi-objective approach for MA so far~\cite{QJK17} shows that any method for multi-objective MDP can be applied on (a discretized version of) an MA for queries involving unbounded or time-bounded reachability probabilities and expected total rewards, but no long-run average rewards.

\section{Preliminaries}
\label{sec:prelim}

The set of \emph{probability distributions} over a finite set $\Omega$ is given by $\dists{\Omega} = \{ \dist \colon \Omega \mapsto [0,1] \mid \sum_{\omega \in \Omega} \dist(\omega) = 1\}$.
For a distribution $\dist \in \dists{\Omega}$ we let $\supp{\dist} = \set{ \omega \in \Omega \mid \dist(\omega) > 0  }$ denote its support. $\dist$ is \emph{Dirac} if $|\supp{\dist}| = 1$.

Let $\rrnn = \set{x \in \rr \mid x \ge 0}$, $\rrgz = \set{x \in \rr \mid x > 0}$, and $\rrext = \rr \cup \set{-\infty, \infty}$ denote the non-negative, positive, and extended real numbers, respectively.
For a point $\point = \tuple{\pointi_1, \dots, \pointi_\numobj} \in \pointdomain$, $\numobj \in \nn$ and $i \in \set{1, \dots, \numobj}$ we write $\tupleaccess{\point}{i}$ for its $i^\mathrm{th}$ entry $\pointi_i$. 
For $\point,\vec{q} \in \pointdomain$ let $\point \cdot \vec{q}$ denote the dot product. We further write $\point \le \vec{q}$ iff $\fa{i} \tupleaccess{\point}{i} \le \tupleaccess{\vec{q}}{i}$ and $\point \lneq \vec{q}$ iff $\point \le \vec{q} \wedge \point \neq \vec{q}$.
The \emph{closure} of a set $\points \subseteq \pointdomain$ is the union of $\points$ and its boundary, denoted by $\cl{\points}$.
The \emph{convex hull} of $\points$ is given by $\conv{P} = \set{ \sum_{i = 1}^{\numobj} \dist(i) \cdot \point_i \mid \dist \in \dists{\set{1, \dots, \numobj}}, \point_1, \dots, \point_\numobj \in \points }$.
The \emph{downward convex hull} of $\points$ is given by $\downset{\points} = \set{\vec{q} \in \pointdomain \mid \ex{\point \in \conv{\points}} \vec{q} \le \point}$.

\subsection{Markov Automata}
Markov automata~(MA)~\cite{EHZ10,DH13} provide an expressive formalism that allows one to model exponentially distributed delays, nondeterminism, probabilistic branching, and instantaneous (undelayed) transitions.
\begin{definition}
\label{def:ma}
A \emph{Markov Automaton} is a tuple $\ma = \matuple$ where
$\states$ is a finite set of states, $\actions$ is a finite set of actions,
$\transitions \colon \states \to \rrgz \cup 2^{\actions}$ is a transition function assigning exit rates to Markovian states $\ms[\ma] = \set{\state \in \states \mid \transitions(\state) \in \rrgz}$ and sets of enabled actions to probabilistic states $\ps[\ma] = \set{\state \in \states \mid \transitions(\state) \subseteq \actions}$, and
$\probabilities \colon \ms[\ma] \cup \sa[\ma] \to \dists{\states}$ with $\sa[\ma] = \set{\sapair \in \ps \times \actions \mid \action \in \transitions(\state)}$ is a  probability function that assigns a distribution over possible successor states for each Markovian state and enabled state-action pair.
\end{definition}
Let $\ma = \matuple$ be an MA.
If $\ma$ is clear from the context, we may omit the superscript from $\ms[\ma]$, $\ps[\ma]$, $\sa[\ma]$, and further notations introduced below. 
Intuitively, the time $\ma$ stays in a Markovian state $\mstate \in \ms$ is governed by an \emph{exponential distribution} with rate $\transitions(\mstate) \in \rrgz$, \ie the probability to take a transition from $\mstate$ within $\dur \in \rrnn$ time units is $1-e^{-\transitions(\mstate) \cdot \dur}$.
Upon taking a transition, a successor state $\state' \in \states$ is drawn from the distribution $\probabilities(\mstate)$, \ie $\probabilities(\mstate)(\state')$ is the probability that the transition leads to $\state' \in \states$.
For probabilistic states $\pstate \in \ps$, an enabled action $\action \in \transitions(\pstate)$ has to be picked and a successor state is drawn from $\probabilities(\psapair)$ (without any delay).
Nondeterminism is thus only possible at probabilistic states.
We assume deadlock free MA, \ie $\fa{\state \in \ps[\ma]} \transitions(\state) \neq \emptyset$.
\begin{remark}
	To enable more flexible modeling such as parallel compositions, the literature (\eg \cite{EHZ10,GTHRS14}) often considers a more liberal variant of MA where (i) different successor distributions can be assigned to the same state-action pair and (ii) states can be both, Markovian \emph{and} probabilistic.
MAs as in \Cref{def:ma}---also known as closed MA---are equally expressive: they can be constructed via action renaming and by applying the so-called \emph{maximal progress assumption}~\cite{EHZ10}.
\end{remark}
An \emph{infinite path }in $\ma$ is a sequence $\infpath = \state_0 \actionordur_1 \state_1 \actionordur_2 \dots$ where for each $i \ge 0$ either
$\state_i \in \ms$, $\actionordur_{i+1} \in \rrnn$, and $\probabilities(\state_i)(\state_{i+1}) > 0$ or
$\state_i \in \ps$, $\actionordur_{i+1} \in \transitions(\state_i)$, and $\probabilities(\tuple{\state_i, \actionordur_{i+1}})(\state_{i+1}) > 0$.
Intuitively, if $\state_i$ is Markovian, $\actionordur_{i+1} \in \rrnn$ reflects the time we have stayed in $\state_i$ until transitioning to $\state_{i+1}$.
If $\state_i$ is probabilistic, $\actionordur_{i+1} \in \actions$ is the performed action via which we transition to $\state_{i+1}$.
A finite path $\finpath = \state_0 \actionordur_1 \state_1 \actionordur_2 \dots \actionordur_{n} \state_n$ is a finite prefix of an infinite path $\infpath$.
We set $\last{\finpath} = \state_n$ and $\lengthofpath{\finpath} = n$ for finite $\finpath$ and $\lengthofpath{\infpath} = \infty$ for infinite $\infpath$.
For (finite or infinite) path $\finorinfpath = \state_0 \actionordur_1 \state_1 \actionordur_2 \dots$ let 
$\durof{\finorinfpath} = \sum_{i=1}^{\lengthofpath{\finorinfpath}} \durof{\actionordur_i}$ be the total duration of $\finorinfpath$ where $\durof{\actionordur} = \actionordur$ if $\actionordur \in \rrnn$ and $0$ otherwise.
If $\finorinfpath$ is infinite and $\durof{\finorinfpath} < \infty$, the path is called \emph{Zeno}.
For $\stepbound \in \nn$ with $\stepbound \le \lengthofpath{\finorinfpath}$ we let $\prefixsteps{\finorinfpath}{\stepbound}$ denote the unique prefix $\pi'$ of $\finorinfpath$ with $\lengthofpath{\pi'} = \stepbound$
and for $\dur \in \rrnn$ we let $\prefixtime{\finorinfpath}{\dur}$
denote the largest prefix of $\finorinfpath$ with total duration at most $\dur$.
The sets of infinite and finite paths of $\ma$ are given by $\infpaths{\ma}$ and $\finpaths{\ma}$, respectively.

A \emph{component} of $\ma$ is a set $\component \subseteq \ms \cup \sa$.
We set 
$\statesof{\component} = (\component \cap \ms) \cup \set{\state \in \ps  \mid \ex{\action} \sapair \in \component}$.
$\component$ is \emph{closed} if $\fa{\componentelem \in \component} \supp{\probabilities(\componentelem)} \subseteq \statesof{\component}$ and  \emph{connected} if for all $\state,\state' \in \statesof{\component}$ there is $\state_0 \actionordur_1 \dots \actionordur_{n} \state_n \in \finpaths{}$ with $\state = \state_0$, $\state' = \state_n$, and for each $i\ge0$ either $\state_i \in \component \cap \ms$ or $\tuple{\state_i,\actionordur_{i+1}} \in \component \cap \sa$.
An \emph{end component (EC)} is a closed and connected component.
An EC is \emph{maximal} if it is not a proper subset of another EC.
$\mecs{\ma}$ denotes the maximal ECs of~$\ma$.
For an EC $\component$ let $\exitsof{\component} = \set{\tuple{\state, \action} \in \sa[\ma] \mid \state \in \statesof{\component} \text{ and } \tuple{\state, \action} \notin \component }$.
\begin{definition}\label{def:subma}
The \emph{sub-MA} of $\ma$ induced by a closed component $\component$ is given by $\submodel{\ma}{\component} = \tuple{\statesof{\component}, \actions, \transitions_\component, \probabilities_\component}$ where $\transitions_\component(\state) = \transitions(\state)$ 
if $\state \in  \component \cap \ms[\ma]$ and otherwise
$\transitions_\component(\state) = \set{ \action \in \transitions(\state) \mid \sapair \in \component}$, and $\probabilities_\component$ is the restriction of $\probabilities$ to~$\component$.
\end{definition}
A \emph{strategy} for $\ma$ resolves the nondeterminism at probabilistic states by providing probability distributions over enabled actions based on the execution history.
\begin{definition}\label{def:strat}
A (general) \emph{strategy} for MA $\ma = \matuple$ is a function $\strat \colon \finpaths{} \to \dists{\actions} \cup \set{\actbot}$ such that for $\finpath \in\finpaths{}$ we have
$\strat(\finpath) \in \dists{\transitions(\last{\finpath})}$ if $\last{\finpath} \in \ps$  and $\strat(\finpath) = \actbot$ otherwise.
\end{definition}
A strategy $\strat$ is called
 \emph{memoryless} if the choice only depends on the current state, \ie
$\fa{\finpath, \finpath' \in \finpaths{}} \last{\finpath} = \last{\finpath'} \text{ implies } \strat(\finpath) = \strat(\finpath')$.
If all assigned distributions are Dirac, $\strat$ is called \emph{deterministic}. 
Let $\strats{\ma}$ and $\stratsmd{\ma}$ denote the set of general and memoryless deterministic strategies of $\ma$, respectively. For simplicity, we often interpret $\strat \in \stratsmd{\ma}$ as a function $\strat \colon \states \to \actions \cup \set{\actbot}$.
The \emph{induced sub-MA} for $\strat \in \stratsmd{\ma}$ is given by $ \submodel{\ma}{\,\ms \cup \set{\tuple{\state, \strat(\state)} \mid \state \in \ps}\,}$.
Strategy $\strat \in \strats{\ma}$ and initial state $\sinit \in \states$ define a \emph{probability measure} $\probmeasure{\ma,\sinit}{\strat}$ that assigns probabilities to sets of infinite paths~\cite{HH12}.
The expected value of $\valuefunc \colon \infpaths{} \to \rrext$ is given by the Lebesque integral
$
\expval{\ma,\sinit}{\strat}(\valuefunc) = \int_{\infpath \in \infpaths{}} \valuefunc(\infpath) \diff \probmeasure{\ma,\sinit}{\strat}(\infpath)
$.

\subsection{Reward-based Objectives}

MA can be equipped with \emph{rewards} to model various quantities like, \eg energy consumption or the number of produced units.
We distinguish between 
\emph{transition} rewards $\rewardassignment_\mathrm{trans} \colon \ms \cup \sa \times \states \to \rr$ that are collected when transitioning from one state to another and
\emph{state} rewards $\rewardassignment_\mathrm{state} \colon \states \to \rr$ that are collected over time, \ie staying in state $\state$ for $\dur$ time units yields a reward of $\rewardassignment_\mathrm{state}(\state) \cdot \dur$.
Since no time passes in probabilistic states, state rewards $\rewardassignment_\mathrm{state}(\state)$ for $\state \in \ps$ are not relevant.
A reward assignment combines the two notions.
\begin{definition}
A \emph{reward assignment} for MA $\ma$ and $\rewardassignment_\mathrm{state}, \rewardassignment_\mathrm{trans}$ as above is a function $\rewardassignment \colon (\ms \times \rrnn) \cup \sa \times \states \to \rr$ with
\[
\rewardassignment(\tuple{\state,\actionordur},\state') =
\begin{cases}
	\rewardassignment_\mathrm{state}(\state) \cdot \actionordur + \rewardassignment_\mathrm{trans}(\state,\state') & \text{if } \state \in \ms, \actionordur \in \rrnn\\
	\rewardassignment_\mathrm{trans}(\tuple{\state,\actionordur},\state') & \text{if } \state \in \ps, \actionordur \in \transitions(\state).
\end{cases}
\]
\end{definition}
We fix a reward assignment $\rewardassignment$ for $\ma$.
$\rewardassignment$ can also be applied to any sub-MA $\submodel{\ma}{\component}$ of $\ma$ in a straightforward way.
%
For a component $\component \subseteq \ms \cup \sa$ we write $\rewardassignment(\component) \ge 0$ if all rewards assigned within $\component$ are non-negative, formally $\fa{\tuple{\state,\actionordur} \in (\component \cap \sa) \cup ((\component \cap \ms) \times \rrnn)} \fa{\state' \in \statesof{\component}} \rewardassignment(\tuple{\component,\actionordur},\state') \ge 0$.
The shortcuts $\rewardassignment(\component) \le 0$ and $\rewardassignment(\component) = 0$ are similar.
The reward of a finite path $\finpath = \state_0 \actionordur_1 \state_1 \actionordur_2 \dots \actionordur_{n} \state_n$ is denoted by $\rewofpath{\rewardassignment}{\finpath} = \sum_{i=1}^{\lengthofpath{\finorinfpath}} \rewardassignment(\tuple{\state_{i-1}, \actionordur_i}, \state_i)$.
%
\begin{definition}\label{def:tot}
The \emph{total reward objective} for reward assignment $\rewardassignment$ is given by  $\totfunc{\rewardassignment}\colon \infpaths{} \to \rrext$ with
$\totfunc{\rewardassignment}(\infpath) = \limsup_{\stepbound \to \infty}  \rewofpath{\rewardassignment}{\prefixsteps{\infpath}{\stepbound}}$.
\end{definition}
\begin{definition}\label{def:lra}
The \emph{long-run average (LRA) reward objective} for $\rewardassignment$ is given by $\lrafunc{\rewardassignment}\colon \infpaths{} \to \rrext$ with
$\lrafunc{\rewardassignment}(\infpath) = \limsup_{\dur\to \infty} \frac{1}{ \dur} \cdot \rewofpath{\rewardassignment}{\prefixtime{\infpath}{\dur}}$.
\end{definition}
\Cref{sec:weighted} considers assumptions under which the limit in both definitions can be attained, \ie $\limsup$ can be replaced by $\lim$.
%
The incorporation of other objectives such as \emph{reachability probabilities} are discussed in \Cref{rem:others}.

\subsection{Markov Decision Processes}
A \emph{Markov Decision Process (MDP)} $\mdp$  is an MA with only probabilistic states, \ie $\ms[\mdp] = \emptyset$.
All notions above also apply to MDP.
However, since all  paths of an MDP have duration 0, there is no timing information available.
For MDP, we therefore usually consider \emph{steps} instead of time.
In particular, for reward assignment $\rewardassignment$ we consider $\lrafuncsteps{\rewardassignment}$ instead of $\lrafunc{\rewardassignment}$, where
$\lrafuncsteps{\rewardassignment}(\infpath) = \limsup_{\stepbound\to \infty} \frac{1}{ \stepbound} \cdot \rewofpath{\rewardassignment}{\prefixsteps{\infpath}{\stepbound}}$.
%
Below, we focus on MA.
Applying our results to step-based LRA rewards on MDPs is straightforward.
Time-based LRA reward objectives for MA can not straightforwardly be reduced to step-based measures for MDP due to the interplay of delayed- and undelayed transitions.

\section{Efficient Multi-objective Model Checking}\label{sec:multi}

We formalize common tasks in multi-objective model checking and sketch our solution method based on~\cite{FKP12}.
We fix an MA $\ma= \matuple$  with initial state $\sinit \in \states$ and $\numobj >0$ objectives $\valuefunc_1, \dots, \valuefunc_\numobj \colon \infpaths{} \to \rr$ with $\multiobjquery = \tuple{\valuefunc_1, \dots, \valuefunc_\numobj}$.
The notation for expected values is lifted to tuples: $\expval{}{\strat}(\multiobjquery) = \tuple{\expval{}{\strat}(\valuefunc_1), \dots,  \expval{}{\strat}(\valuefunc_\numobj)}$.

\subsection{Multi-objective Model Checking Queries}
Our aim is to maximize the expected value for each (potentially conflicting) objective $\valuefunc_\objindex$.
We impose the following assumption which can be asserted using single-objective model checking.
We further discuss the assumption in \Cref{rem:finiteness}. 
\begin{assumption}[Objective Finiteness] \label{as:finite}
	$\fa{\objindex} \sup \set{ \expval{}{\strat}(\valuefunc_\objindex) \mid {\strat \in \strats{}}} < \infty
	$.
\end{assumption}
\begin{definition}\label{def:achpar}
For $\multiobjquery$ as above, 
$
\ach{}{\multiobjquery} = \set{ \point \in \pointdomain \mid \ex{\strat \in \strats{}} \point \le \expval{}{\strat}(\multiobjquery) }
$
is the set of \emph{achievable points}. The \emph{Pareto front} is given by
$
\pareto{}{\multiobjquery} = \set{ \point \in \cl{\ach{}{\multiobjquery}} \mid \fa{\point' \gneq \point} \point' \notin \cl{\ach{}{\multiobjquery}} }.
$
\end{definition}
%
A point $\point \in \ach{}{\multiobjquery}$ is called \emph{achievable} and there is a single strategy $\strat$ that for each objective $\valuefunc_\objindex$ \emph{achieves} an expected value of at least $\tupleaccess{\point}{\objindex}$.
Due to \Cref{as:finite}, the Pareto front is the \emph{frontier} of the set of achievable points, meaning that it is the smallest set $\points \subseteq \pointdomain$ with $\downset{\points} = \cl{\ach{}{\multiobjquery}}$.
We can thus interpret  $\pareto{}{\multiobjquery}$ as a representation for $\cl{\ach{}{\multiobjquery}}$ and vice versa.
The set of achievable points is closed iff all points on the Pareto front are achievable. 

\begin{figure}[t]
	\centering
	\begin{subfigure}[b]{0.51\linewidth}
			\centering
			\newcommand{\trate}[1]{{\scriptsize\ensuremath{\mathbf{#1}}}}
			\newcommand{\tprob}[1]{{\scriptsize\ensuremath{#1}}}
			\newcommand{\tact}[1]{{\scriptsize\ensuremath{#1}}}
			\newcommand{\trew}[1]{{\scriptsize\ensuremath{#1}}}
			\begin{tikzpicture}[->,node
				distance=0.8cm,bend angle=45, scale=0.6]
				\node[ms] (1) at (0,0) {$\state_1$};
				\node[ms, right=2.5cm of 1,label=0:\trew{\tuple{6,0}}] (2)  {$\state_2$};
				\node[ps, below=of 1] (3)  {$\state_3$};
				\node[ps, below=of 2] (4)  {$\state_4$};
				\node[ms, below=of 3,label=180:\trew{\tuple{2,0}}] (5)  {$\state_5$};
				\node[ms, below=of 4,label=0:\trew{\tuple{1,0}}] (6)  {$\state_6$};
				
				\path[->,densely dotted, semithick]
					(1) edge[] node[pos=0.25,above] {\trate{1}} node[dist] (d1) {} node[pos=0.75,above] {\tprob{0.5}} (2)
					(d1) edge[bend left=15] node[pos=0.5,right] {\tprob{0.5}} (3)
					(2) edge[] node[pos=0.25,right] {\trate{2}} node[dist] (d2) {} node[pos=0.75,right] {\tprob{1}} (4)
					(5) edge[loop above] node[pos=0.25,left] {\trate{1}} node[dist,yshift=-1pt] (d5) {} node[pos=0.75,right] {\tprob{1}} (5)
					(6) edge[] node[pos=0.25,right] {\trate{2}} node[dist] (d6) {} node[pos=0.75,right] {\tprob{1}} (4)
				;
				\path[->,semithick]
					($(3.west) - (1,0)$) edge[] (3) 
					(3) edge[] node[pos=0.25,left] {\tact{\alpha}} node[dist] (d3a) {} node[pos=0.75,left] {\trew{\tuple{1,-1}}} node[pos=0.75,right] {\tprob{1}} (1)
					(5) edge[<->] node[pos=0.25,above] {\tprob{0.5}} node[dist] (d3b) {} node[pos=0.75,above] {\tprob{0.5}} (6)
					(3) edge[-,bend left=15] node[pos=0.5,right] {\tact{\beta}} (d3b)
					
					(2) edge[<->,bend right] node[pos=0.25,left] {\tprob{0.6}} node[pos=0.5,dist] (d4a) {} node[pos=0.75,left] {\tprob{0.4}} (6)
					(4) edge[-] node[pos=0.5,above] {\tact{\alpha}} (d4a)
					
					(2) edge[<->,bend left] node[pos=0.25,right] {\tprob{0.3}} node[pos=0.5,dist] (d4b) {} node[pos=0.75,right] {\tprob{0.7}} (6)
					(4) edge[-] node[pos=0.5,above] {\tact{\beta}} (d4b)
					
				;
			\end{tikzpicture}
		\caption{MA $\ma$ with rewards $\tuple{\rewardassignment_1, \rewardassignment_2}$}
			\label{fig:ma}
	\end{subfigure}
\begin{subfigure}[b]{0.48\linewidth}
	\centering
	\begin{tikzpicture}[scale=1]
		\begin{axis}[
			width=5cm,
			xlabel={$\lrafunc{\rewardassignment_{1}}$},
			xmin=0,
			xmax = 5,
			xtick={0,1,2,3,4,5},
			ylabel={$\totfunc{\rewardassignment_{2}}$},
			ylabel style={yshift=-0.4cm},
			xlabel style={yshift=0.3cm},
			yticklabel style={font=\scriptsize},
			xticklabel style={font=\scriptsize},
			ymin=-3,
			ytick={-3,-2,-1,0,1},
			ymax=1.25,
			axis background/.style={fill=red!50},
			axis on top
			]
			
			\addplot[fill=green!50, very thin] coordinates {(-1,-4) (-1,0) (3,0) (4,-2) (4,-4)} -- cycle;
			
			\addplot[blue, ultra thick] coordinates  {(3,0) (4,-2)} ;
			\node[point,label=0:$\point_1$] at (axis cs:4,-2) {};
			\node[point,label=90:$\point_2$] at (axis cs:3,0) {};
			\node[] at (axis cs:2,-1.5) {$\ach{}{\multiobjquery}$};
			
		\end{axis} 
	\end{tikzpicture}
		\caption{$\ach{}{\multiobjquery}$ (green) and $\pareto{}{\multiobjquery}$ (blue)}
\label{fig:ach}
\end{subfigure}
		\caption{MA with achievable points and Pareto front for $\multiobjquery = \tuple{\lrafunc{\rewardassignment_{1}}, \totfunc{\rewardassignment_{2}} }$}
		\label{fig:maach}
\end{figure}

\begin{example}
 \Cref{fig:ma} shows an MA with initial state $\state_3$.
	Transitions are annotated with actions, rates (boldfaced), and successor probabilities.
	We also depict two reward assignments $\rewardassignment_{1}$ and $\rewardassignment_{2}$ by labeling states and transitions with tuples $\tuple{r_1,r_2}$ where, \eg
	$\rewardassignment_{2}(\state_3, \alpha, \state_1) = -1$ and for $\dur \in \rrnn$: $\rewardassignment_{1}(\state_2,\dur,\state_4) =  6 \cdot \dur$.

	For $\strat_1 \in \stratsmd{}$ with $\strat_1\colon \state_3,\state_4 \mapsto \alpha$, the EC $\set{\state_2, \tuple{\state_4,\alpha}, \tuple{\state_4,\beta}, \state_6}$ is reached almost surely (with probability 1), yielding $\expval{}{\strat_1}(\lrafunc{\rewardassignment_{1}}) = 0.6 \cdot 6 + 0.4 \cdot 1 = 4$ 
	and $\expval{}{\strat_1}(\totfunc{\rewardassignment_{2}}) = \sum_{i=0}^{\infty} -1 \cdot (0.5)^i = -2$.
	It follows that the point $\point_1 = \tuple{4,-2}$ as indicated in \Cref{fig:ach} is achievable.
	Similarly, $\strat_2 \in \stratsmd{}$ with $\strat_2 \colon \state_3 \mapsto \beta, \state_4 \mapsto \alpha$ achieves the point $\point_2 = \tuple{3,0}$.
	With strategies that randomly pick an action at $\state_3$, we can also achieve any point on the blue line in~\Cref{fig:ach} that connects $\point_1$ and $\point_2$. 
	This line coincides with the Pareto front $\pareto{}{\multiobjquery}$ for $\multiobjquery = \tuple{\lrafunc{\rewardassignment_{1}}, \totfunc{\rewardassignment_{2}} }$.
	The set of achievable points $\ach{}{\multiobjquery}$ (indicated in green) coincides with the downward convex hull of the Pareto front.
\end{example}
For multi-objective model checking we are concerned with the following queries: 
\begin{nproblem}[framed]{Multi-objective Model Checking Queries}
Qualitative Achievability:& Given point $\point \in \pointdomain$, decide if $\point \in \ach{}{\multiobjquery}$.\\
Quantitative Achievability:& Given $\pointi_2, \pointi_3, \dots, \pointi_\numobj \in \pointidomain$, compute or approximate
$\sup \set{\pointi \in \pointidomain \mid \tuple{\pointi, \pointi_2, \pointi_3, \dots, \pointi_\numobj} \in \ach{}{\multiobjquery}}$.\\
Pareto:& Compute or approximate $\pareto{}{\multiobjquery}$.
\end{nproblem}
%

\subsection{Approximation of Achievable Points}
A practically efficient approach that tackles the above queries for expected total rewards in MDP was given in \cite{FKP12}.
It is based on so-called \emph{sandwich algorithms} known from convex multi-objective optimization~\cite{SAC93,RDH11}.
We extend the algorithm to arbitrary combinations of objectives $\valuefunc_\objindex$ on MA, including---and this is the main algorithmic novelty---mixtures of total- and LRA reward objectives.

\begin{algorithm}[t]
	\Input{MA $\ma$ with initial state $\sinit$, objectives $\multiobjquery = \tuple{\valuefunc_1, \dots, \valuefunc_\numobj}$}
	\Output{An approximation of $\ach{}{\multiobjquery}$}
	
	$\points \gets \emptyset$ \tcp{Collects achievable points found so far.}
	$\overapprox \gets \pointdomain$ \tcp{Excludes points that are known to be unachievable.}
	\DoUntil{Approximation $\downset{\underapprox} \subseteq \ach{}{\multiobjquery} \subseteq \overapprox$ answers multi-obj. query}{\label{line:refineloop:start}
		Select weights $ \weightvector \in \{ \weightvector' \in (\rrnn)^\numobj \mid \sum_{\objindex = 1}^{\numobj} \tupleaccess{\weightvector'}{\objindex} = 1\}$ and $\varepsilon > 0$\\
		Find $ v_\weightvector \ge \sup \set{ \weightvector \cdot \expval{}{\strat}(\multiobjquery) \mid \strat \in \strats{}} $, $ \strat_\weightvector \in \strats{}$ s.t. $ |v_\weightvector - \weightvector \cdot \expval{}{\strat_\weightvector}(\multiobjquery)| \le \varepsilon$ \label{alg:line:upperboundandstrat}\\
		Compute $\point_\weightvector \in \pointdomain$ with $\fa{\objindex}  \tupleaccess{\point_\weightvector}{\objindex} = \expval{}{\strat_\weightvector}(\valuefunc_\objindex)$ \label{alg:line:inducedpoint}\\
		$\underapprox \gets \underapprox \cup \set{\point_\weightvector}$; $\overapprox \gets \overapprox \cap \set{ \point \in \pointdomain \mid \weightvector \cdot \point \le v_\weightvector}$\label{alg:line:extend}
	}
	\caption{Approximating the set of achievable points}
	\label{alg:refinementloop}
\end{algorithm}

The idea is to iteratively refine an approximation of the set of achievable points $\ach{}{\multiobjquery}$.
The refinement loop is outlined in \Cref{alg:refinementloop}.
At the start of each iteration, the algorithm chooses a weight vector $\weightvector$ and a precision parameter $\varepsilon$ after some heuristic (details below).
Then, \Cref{alg:line:upperboundandstrat}, considers the weighted sum of the expected values of the objectives $\valuefunc_\objindex$.
More precisely, an upper bound $v_\weightvector$ for $\sup \set{ \weightvector \cdot \expval{}{\strat}(\multiobjquery) \mid \strat \in \strats{}}$ as well as a ``near optimal'' strategy $\strat_\weightvector$ need to be found such that the difference between the bound $v_\weightvector$ and the weighted sum induced by $\strat_\weightvector$ is at most $\varepsilon$.
In \Cref{sec:weighted}, we outline the computation of $v_\weightvector$ and $\strat_\weightvector$ for the case where $\multiobjquery$ consists of total-and LRA reward objectives.
Next, in \Cref{alg:line:inducedpoint} the algorithm computes a point $\point_\weightvector$ that contains the expected values for each individual objective $\valuefunc_\objindex$ under strategy $\strat_\weightvector$.
These values can be computed using off-the-shelf single-objective model checking algorithms on the model induced by $\strat_\weightvector$.
By definition, $\point_\weightvector$ is achievable.
Finally, \Cref{alg:line:extend} inserts the found point into the initially empty set $\underapprox$ and excludes points from the set $\overapprox$ (which initially contains all points) that are known to be unachievable.
The following theorem establishes the correctness of the approach.
We prove it using \Cref{lem:conv,lem:supremumbounds}.
\begin{theorem}\label{thm:algcorrect}
\Cref{alg:refinementloop} maintains the invariant $\downset{\underapprox} \subseteq \ach{}{\multiobjquery} \subseteq \overapprox$.
\end{theorem}
\begin{lemma}\label{lem:supremumbounds}
	$\fa{\point \in \ach{}{\multiobjquery}, \weightvector \in (\rrnn)^\numobj} \weightvector \cdot \point \le \sup \set{\weightvector \cdot \expval{}{\strat}(\multiobjquery) \mid \strat \in \strats{}}$.
\end{lemma}
\begin{proof}
	Let $\point \in \ach{}{\multiobjquery}$ be achieved by strategy $\strat_\point \in \strats{}$. The claim follows from
	$$
	\weightvector \cdot \point 
	= \sum_{\objindex = 1}^{\numobj} \tupleaccess{\weightvector}{\objindex} \cdot \tupleaccess{\point}{\objindex}
	\le \sum_{\objindex = 1}^{\numobj} \tupleaccess{\weightvector}{\objindex} \cdot \expval{}{\strat_\point}(\valuefunc_\objindex)
	\le \sup \Big\{ \sum_{\objindex = 1}^{\numobj} \tupleaccess{\weightvector}{\objindex} \cdot  \expval{}{\strat}(\valuefunc_\objindex) \,\Big|\, \strat \in \strats{}\Big\}.
	$$
\end{proof}
\begin{lemma}\label{lem:conv}
$\ach{}{\multiobjquery}$ is convex, \ie $\ach{}{\multiobjquery} = \conv{\ach{}{\multiobjquery}}$.
\end{lemma}
\begin{proof}
	We need to show that for two points $\point_1,\point_2 \in \ach{}{\multiobjquery}$ with achieving strategies $\strat_1, \strat_2 \in \strats{}$, any point $\point$ on the line connecting $\point_1$ and $\point_2$ is also achievable.
	Formally, for $w \in [0,1]$ show that $\point_w = w \cdot \point_1 + (1{-}w) \cdot \point_2 \in \ach{}{\multiobjquery}$.
	Consider the strategy $\strat_w$ that initially makes a coin flip%
	\footnote{Strategies as in \Cref{def:strat} can not ``store'' the outcome of the initial coin flip.
	Thus, given $\finpath \in \finpaths{}$, strategy $\strat_w$ actually has to consider the \emph{conditional} probability for the outcome of the coin flip, given that $\finpath$ has been observed. Alternatively, we could have also introduced strategies with memory.}%
	: With probability $w$ it mimics $\strat_1$ and otherwise it mimics $\strat_2$.
	We can show for all objectives $\valuefunc_\objindex$:
	$$
	\tupleaccess{\point_\weightvector}{\objindex}
	= w \cdot \tupleaccess{\point_1}{\objindex} + (1{-}w) \cdot \tupleaccess{\point_2}{\objindex}
	\le w \cdot \expval{}{\strat_1}(\valuefunc_\objindex) + (1{-}w) \cdot \expval{}{\strat_2}(\valuefunc_\objindex)
	= \expval{}{\strat_w}(\valuefunc_\objindex).
	$$
\end{proof}
We now show \Cref{thm:algcorrect}. A similar proof was given in~\cite{FKP12}.
\begin{proof}[of \Cref{thm:algcorrect}]
	All $\point_\weightvector \in \underapprox$ are achievable, \ie $\underapprox \subseteq \ach{}{\multiobjquery}$. By \Cref{def:achpar} and \Cref{lem:conv} we get $\downset{\underapprox} \subseteq \downset{\ach{}{\multiobjquery}} = \conv{\ach{}{\multiobjquery}} = \ach{}{\multiobjquery}$.
	Now let $\point \in \ach{}{\multiobjquery}$ and let $\weightvector$ be an arbitrary weight vector considered in some iteration of \Cref{alg:refinementloop} with corresponding value $v_\weightvector$ computed in \Cref{alg:line:upperboundandstrat}. \Cref{lem:supremumbounds} yields $\weightvector \cdot \point \le \sup \set{ \weightvector \cdot \expval{}{\strat}(\multiobjquery) \mid \strat \in \strats{}} \le  v_\weightvector $ and thus $\point \in \overapprox$.
\end{proof}
\Cref{alg:refinementloop} can be stopped at any time and the current approximation of $\ach{}{\multiobjquery}$ can be used to (i) decide qualitative achievability, (ii) provide a lower and an upper bound for quantitative achievability, and (iii) obtain an approximative representation of the Pareto front.

The \emph{precision parameter} $\varepsilon$ can be decreased dynamically to obtain a gradually finer approximation.
If $\ach{}{\multiobjquery}$ is closed, the supremum $\sup \set{\weightvector \cdot \expval{}{\strat}(\multiobjquery) \mid \strat \in \strats{}}$ can be attained by some strategy $\strat_\weightvector$, allowing us to set $\varepsilon = 0$.

We briefly sketch the \emph{selection of weight vectors} as proposed in~\cite{FKP12}.
In the first $\numobj$ iterations of \Cref{alg:refinementloop}, we optimize each objective $\valuefunc_\objindex$ individually, \ie we consider for all $\objindex$ the weight vector $\weightvector$ with $\tupleaccess{\weightvector}{i} = 0$ for $i \neq \objindex$ and $\tupleaccess{\weightvector}{\objindex} = 1$.
After that, we consider weight vectors that are orthogonal to a facet of the downward convex hull of the current set of points $\underapprox$.
To approximate the Pareto front, facets with a large distance to $\pointdomain \setminus \overapprox$ are considered first.
To answer a qualitative or quantitative achievability query, the selection can be  guided further based on the input point $\point \in \pointdomain$ or the input values $\pointi_2, \pointi_3, \dots, \pointi_\numobj \in \pointidomain$.
More details and further discussions on these heuristics can be found in~\cite{FKP12}.


\begin{remark}\label{rem:finiteness}
\Cref{as:finite} does not exclude $\expval{}{\strat}(\valuefunc_\objindex) = -\infty$ which occurs, \eg when objectives reflect resource consumption and some (bad) strategies require infinite resources.
Moreover, if \Cref{as:finite} is violated for an objective $\valuefunc_\objindex$ 
 we observe that for this objective, any (arbitrarily high) value $\pointi \in \rr$ can be achieved with some strategy $\strat \in \strats{}$ such that $\pointi \le \expval{}{\strat}(\valuefunc_\objindex)$.
Similar to the proof of~\Cref{lem:conv}, a strategy can be constructed that---with a small probability---mimics a strategy inducing a very high expected value for $\valuefunc_\objindex$ and---with the remaining (high) probability---optimizes for the other objectives.
Let $\multiobjquery_{-\objindex}$ be the tuple $\multiobjquery$ without $\valuefunc_\objindex$ and similarly for $\point \in \pointdomain$ let $\point_{-\objindex} \in \pointidomain^{\numobj - 1}$ be the point $\point$ without the $\objindex^\mathrm{th}$ entry.
Assuming $\inf \set{ \expval{}{\strat}(\valuefunc_\objindex) \mid {\strat \in \strats{}}} > -\infty$, we can show that $\cl{\ach{}{\multiobjquery}} = \set{\point \in \pointdomain \mid \point_{-\objindex} \in \cl{\ach{}{\multiobjquery_{-\objindex}}}}$.
Put differently, $\cl{\ach{}{\multiobjquery}}$ can be constructed from the achievable points obtained without the objective $\valuefunc_\objindex$.

\end{remark}
\section{Optimizing Weighted Combinations of Objectives}
\label{sec:weighted}
We now analyze weighted sums of expected values as in \Cref{alg:line:upperboundandstrat} of \Cref{alg:refinementloop}.
%
\begin{nproblem}[framed]{Weighted Sum Optimization Problem}
		Input: & MA $\ma$ with initial state $\sinit$, objectives $\multiobjquery = \tuple{\valuefunc_1, \dots, \valuefunc_\numobj}$,
		\\& weight vector $ \weightvector \in \{ \weightvector' \in (\rrnn)^\numobj \mid \sum_{\objindex = 1}^{\numobj} \tupleaccess{\weightvector'}{\objindex} = 1\}$, precision $\varepsilon > 0$\\
		Output: & Value $v_\weightvector \in \rr$, 
		with $v_\weightvector \ge \sup \set{\weightvector \cdot \expval{}{\strat}(\multiobjquery) \mid \strat \in \strats{} } $ and
		\\ &
		strategy $\strat_\weightvector \in \strats{}$ such that $ |v_\weightvector - \weightvector \cdot \expval{}{\strat_\weightvector}(\multiobjquery)| \le \varepsilon$.
\end{nproblem}
\noindent We only consider total- and LRA reward objectives.
\Cref{rem:others} discusses other objectives.
We show that
instead of a weighted sum of the expected values
we can consider weighted sums of the rewards. 
This allows us to combine all objectives into a single reward assignment and then apply single-objective model checking.

\subsection{Pure Long-run Average Queries}
Initially, we restrict ourselves to LRA objectives and show a reduction of the weighted sum optimization problem to a single-objective long-run average reward computation.
As usual for MA~\cite{HH12,GHHKT14,BWH17} we forbid so-called Zeno behavior.
\begin{assumption}[Non-Zenoness]\label{as:zeno}
$\fa{\strat \in \strats{\ma}} \probmeasure{\ma}{\strat}(\set{\infpath  \mid \durof{\infpath} < \infty}) = 0$.
\end{assumption}
The assumption is equivalent to assuming that every EC of $\ma$ contains at least one Markovian state.
If the assumption holds, the limit in \Cref{def:lra} can be attained almost surely (with probability 1) and corresponds to a value $v \in \rr$.
Thus, \Cref{as:finite} for LRA objectives is already implied by \Cref{as:zeno}.
Let $\lramultiobjquery = \tuple{\lrafunc{\rewardassignment_1}, \dots, \lrafunc{\rewardassignment_\numobj}}$ with reward assignments $\rewardassignment_\objindex$.
Moreover, for weight vector $\weightvector$ let $\rewardassignment_\weightvector$ be the reward assignment with
$\rewardassignment_\weightvector(\tuple{\state, \actionordur}, \state') =  \sum_{\objindex = 1}^{\numobj} \tupleaccess{\weightvector}{\objindex} \cdot \rewardassignment_\objindex(\tuple{\state, \actionordur}, \state')$.
\begin{theorem}\label{thm:lra}
$\fa{\strat \in \strats{}}\weightvector \cdot \expval{}{\strat}(\lramultiobjquery) = \expval{}{\strat}(\lrafunc{\rewardassignment_\weightvector}) $. 
\end{theorem}
\begin{proof}
	Due to \Cref{as:zeno} we have for almost all paths $\infpath \in \infpaths{}$ that for all $\objindex \in \set{1, \dots, \numobj}$ the limit $\lim_{\dur\to \infty} \frac{1}{ \dur} \cdot \rewofpath{\rewardassignment_\objindex}{\prefixtime{\infpath}{\dur}}$ exists and
$$
\sum_{\objindex = 1}^{\numobj} \tupleaccess{\weightvector}{\objindex} \cdot	\lrafunc{\rewardassignment_\objindex}(\infpath)
 = \lim_{\dur\to \infty} \frac{1}{ \dur} \cdot \sum_{\objindex = 1}^{\numobj} \tupleaccess{\weightvector}{\objindex} \cdot  \rewofpath{\rewardassignment_\objindex}{\prefixtime{\infpath}{\dur}}
= \lrafunc{\rewardassignment_\weightvector}(\infpath).
 $$
The theorem follows with
$$
\sum_{\objindex = 1}^{\numobj} \tupleaccess{\weightvector}{\objindex} \cdot \expval{}{\strat}(\lrafunc{\rewardassignment_\objindex})
 = \int_{\infpath}\ \sum_{\objindex = 1}^{\numobj} \tupleaccess{\weightvector}{\objindex} \cdot \lrafunc{\rewardassignment_\objindex} \diff \probmeasure{}{\strat}(\infpath)
=\expval{}{\strat}(\lrafunc{\rewardassignment_\weightvector}).
 $$
\end{proof}
Due to \Cref{thm:lra}, it suffices to consider the expected LRA reward for the \emph{single} reward assignment $\rewardassignment_\weightvector$.
The supremum $\sup \set{\expval{}{\strat}(\lrafunc{\rewardassignment_\weightvector}) \mid \strat \in \strats{}}$ is attained by some memoryless deterministic strategy $\strat_\weightvector \in \stratsmd{}$~\cite{GTHRS14}.
Such a strategy and the induced value $v_\weightvector = \expval{}{\strat_\weightvector}(\lrafunc{\rewardassignment_\weightvector})$ 
can be computed (or approximated) with \emph{linear programming}~\cite{GTHRS14}, \emph{strategy iteration}~\cite{KM17} or \emph{value iteration}~\cite{BWH17,ACDKM17}.

\subsection{A Two-phase Approach for Single-objective LRA}
The computation of single-objective expected LRA rewards for reward assignment $\rewardassignment_\weightvector$ can be divided in two phases~\cite{GHHKT14,BWH17,ACDKM17}.
First, each maximal end component $\component \in \mecs{\ma}$ is analyzed individually by computing
for sub-MA $\submodel{\ma}{\component}$ and some\footnote{The value $v_\component$ does not depend on the selected state $\state$. Intuitively, this is because any other state $\state' \in \statesof{\component}$ can be reached from $\state$ almost surely.} $\state \in \statesof{\component}$ the value 
 $v_\component = \max \{ \expval{\submodel{\ma}{\component},\state}{\strat}(\lrafunc{\rewardassignment_\weightvector}) \mid \strat \in \stratsmd{\submodel{\ma}{\component}}\}$.

Secondly, we consider a quotient model $\ma' = \quotient{\ma}{\mecs{\ma}}$ of $\ma$ that replaces the states of each $\component \in \mecs{\ma}$ by a single state.

\begin{definition}\label{def:ecq}
For $\ma = \matuple$ and a set of ECs $\componentset$, the \emph{quotient} is the MA $\quotient{\ma}{\componentset} = \tuple{\quotient{\states}{\componentset}, \quotient{\actions}{\componentset}, \quotient{\transitions}{\componentset},\quotient{\probabilities}{\componentset}}$ where
\begin{itemize}
\item $\quotient{\states}{\componentset} = \left(\states \setminus \bigcup_{\component \in \componentset} \statesof{\component}\right) \uplus \componentset \uplus  \set{\state_{\staysign}}$,
$\quotient{\actions}{\componentset} = \actions \uplus \left(\bigcup_{\component \in \componentset} \exitsof{\component}\right) \uplus \set{\staysign}$,
\item
$\quotient{\transitions}{\componentset}(\hat{\state}) =
\begin{cases}
\transitions(\hat{\state}) & \text{ if } \hat{\state} \in \states \\
\exitsof{\hat{\state}} \cup \set{\staysign} & \text{ if } \hat{\state} \in \componentset\\
1 & \text{ if } \hat{\state} = \state_{\staysign} \text{, and}
\end{cases}$
\item
$\quotient{\probabilities}{\componentset}(\componentelem) =
\begin{cases}
\probabilities(\componentelem) & \text{ if } \componentelem \in \ms[\ma] \cup \sa[\ma]\\
\probabilities(\tuple{\state,\action}) & \text{ if } \componentelem = \tuple{\component, \tuple{\state,\action}} \text{ for } \component \in \componentset \text{ and } \tuple{\state,\action} \in \exitsof{\component}\\
\set{\state_\staysign \mapsto 1} &\text{ if } \componentelem \in \componentset \times \set{\staysign} \cup \set{\state_{\staysign}}
\end{cases}$
\end{itemize}
\end{definition}
Intuitively, selecting action $\staysign$ at a state $\component \in \mecs{\ma}$ in $\ma'$ reflects any strategy of $\ma$ that upon visiting the EC $\component$ will stay in this EC forever.
We can thus mimic any strategy of the sub-MA $\submodel{\ma}{\component}$, in particular a memoryless deterministic strategy that maximizes the expected value of $\lrafunc{\rewardassignment_\weightvector}$ in $\submodel{\ma}{\component}$.
Contrarily, selecting an action $\tuple{\state, \action}$ at a state $\component$ of $\ma'$ reflects a strategy of $\ma$ that upon visiting the EC $\component$ enforces that the states of $\component$ will be left via the exiting state-action pair $\sapair$.
Let $\rewardassignment^\ast$ be the reward assignment for $\ma'$ that yields $\rewardassignment^\ast(\tuple{\component, \staysign}, \state_{\staysign}) = v_\component$ and 0 in all other cases.
It can be shown that
$\max \{\expval{\ma,\sinit}{\strat}(\lrafunc{\rewardassignment_\weightvector}) \mid \strat \in \strats{\ma} \} = \max \{ \expval{\ma',\sinit'}{\strat}(\totfunc{\rewardassignment^\ast}) \mid \strat \in \strats{\ma'} \}$, where $\sinit' = \component_\mathit{I}$ if  $\sinit$ is contained in some $\component_\mathit{I} \in \mecs{\ma}$ and $\sinit' = \sinit$ otherwise.

The maximal total reward in $\ma'$ can be computed using standard techniques such as
\emph{value iteration} and \emph{policy iteration}~\cite{Put94} as well as the more recent \emph{sound value iteration} and \emph{optimistic value iteration}~\cite{QK18,HK20}.
The latter two provide sound precision guarantees for the output value $v$, \ie $|v -  \max \{ \expval{\ma',\sinit'}{\strat}(\totfunc{\rewardassignment^\ast}) \mid \strat \in \strats{\ma'} \}| \le \varepsilon$ for a given $\varepsilon > 0$.

\subsection{Combining Long-run Average and Total Rewards}
\label{sec:lratot}

We now consider arbitrary combinations of total- and long-run average reward objectives $\multiobjquery = \tuple{\totfunc{\rewardassignment_1}, \dots, \totfunc{\rewardassignment_\numtotobj}, \lrafunc{\rewardassignment_{\numtotobj + 1}}, \dots, \lrafunc{\rewardassignment_\numobj} }$ with $0 < \numtotobj < \numobj$.

The above-mentioned procedure for LRA reduces the analysis to an expected total reward computation on the quotient model $\quotient{\ma}{\mecs{\ma}}$.
This approach suggests to also incorporate other total-reward objectives for $\ma$ in the quotient model.
However, special care has to be taken concerning total rewards collected within ECs of $\ma$ that would no longer be present in the quotient $\quotient{\ma}{\mecs{\ma}}$.
We discuss how to deal with this issue by considering the quotient only for ECs in which no (total) reward is collected.
We start with restricting the (total) rewards that might be assigned to transitions within EC.
\begin{assumption}[Sign-Consistency]\label{as:sign}
	For all total reward objectives $\totfunc{\rewardassignment_\objindex}$ either $\fa{\component \in \mecs{\ma}} \rewardassignment_\objindex(\component) \ge 0 $ or $\fa{\component \in \mecs{\ma}} \rewardassignment_\objindex(\component) \le 0 $.
\end{assumption}
The assumption implies that paths on which infinitely many positive \emph{and} infinitely many negative reward is collected have probability 0.
One consequence is that the limit in \Cref{def:tot} exists for almost all paths~\cite{BBDGS18}.
A discussion on objectives $\totfunc{\rewardassignment_\objindex}$ that violate \Cref{as:sign} for single-objective MDP is given in~\cite{BBDGS18}.
Their multi-objective treatment is left for future work.

When \Cref{as:finite,as:sign} hold, we get $\rewardassignment_\objindex(\component) \le 0$ for all objectives $\totfunc{\rewardassignment_i}$ and EC $\component$.
Put differently, all non-zero total rewards collected in an EC have to be negative.
Strategies that induce a total reward of $ -\infty$ for some objective $\totfunc{\rewardassignment_i}$ will not be taken into account for the set of achievable points. Therefore, transitions within ECs that yield negative reward should only be taken finitely often.
These transitions can be disregarded when computing the expected LRA rewards, \ie only the  0-ECs~\cite{BBDGS18} are relevant for the LRA computation.
\begin{definition}\label{def:0ec}
	A 0-EC of $\ma$ and $\rewardassignment_1, \dots, \rewardassignment_\numtotobj$ is an EC $\component$ of $\ma$ with $\rewardassignment_i(\component) = 0$ for all $\rewardassignment_i$.
	The set of maximal 0-ECs is given by $\zeromecs{\ma}{\rewardassignment_1, \dots, \rewardassignment_i}$.
\end{definition}
$\zeromecs{\ma}{\rewardassignment_1, \dots, \rewardassignment_\numtotobj}$ can be computed by constructing the maximal ECs of the sub-MA of $\ma$ where transitions with a non-zero reward are erased.

We are ready to describe our approach that combines LRA rewards of 0-ECs and the remaining total rewards into a single total-reward objective.
Let $\rewardassignment_\weightvector^\mathit{tot}$ and $\rewardassignment_\weightvector^\mathit{lra}$ be reward assignments with
$\rewardassignment_\weightvector^\mathit{tot}(\tuple{\state, \actionordur}, \state') =  \sum_{i = 1}^{\numtotobj} \tupleaccess{\weightvector}{i} \cdot \rewardassignment_i(\tuple{\state, \actionordur}, \state')$ and
$\rewardassignment_\weightvector^\mathit{lra}(\tuple{\state, \actionordur}, \state') =  \sum_{j = \numtotobj}^{\numobj} \tupleaccess{\weightvector}{j} \cdot \rewardassignment_j(\tuple{\state, \actionordur}, \state')$.
Moreover, for $\infpath \in \infpaths{}$ we set $(\totfunc{\rewardassignment_\weightvector^\mathit{tot}} + \lrafunc{\rewardassignment_\weightvector^\mathit{lra}})(\infpath) = \totfunc{\rewardassignment_\weightvector^\mathit{tot}}(\infpath) + \lrafunc{\rewardassignment_\weightvector^\mathit{lra}}(\infpath)$.
\begin{theorem}\label{thm:mixed}
$\fa{\strat \in \strats{}}\weightvector \cdot \expval{}{\strat}(\multiobjquery) = \expval{}{\strat}(\totfunc{\rewardassignment_\weightvector^\mathit{tot}} + \lrafunc{\rewardassignment_\weightvector^\mathit{lra}}) $. 
\end{theorem}
\begin{proof}
Using a similar reasoning as in the proof of \Cref{thm:lra}, we get:
\vspace{-8pt}
\begin{align*}
\weightvector \cdot \expval{}{\strat}(\multiobjquery)
 &= \Big(\sum_{i=1}^{\numtotobj} \tupleaccess{\weightvector}{i} \cdot \expval{}{\strat}(\totfunc{\rewardassignment_i})\Big) + \Big(\sum_{j=\numtotobj + 1}^{\numobj} \tupleaccess{\weightvector}{j} \cdot \expval{}{\strat}(\lrafunc{\rewardassignment_j})\Big) \\
&= \expval{}{\strat}(\totfunc{\rewardassignment_\weightvector^\mathit{tot}}) + \expval{}{\strat}(\lrafunc{\rewardassignment_\weightvector^\mathit{lra}})
= \expval{}{\strat}(\totfunc{\rewardassignment_\weightvector^\mathit{tot}} + \lrafunc{\rewardassignment_\weightvector^\mathit{lra}}).
\end{align*}
\end{proof}
\Cref{alg:weighted} outlines the procedure for solving the weighted sum optimization problem.
It first computes optimal LRA rewards and inducing strategies for each maximal 0-EC (Lines \ref{alg:line:ecstart} to \ref{alg:line:ecend}).
Then, a quotient model $\ma^\ast$ and a reward assignment $\rewardassignment^\ast$ incorporating all total- and LRA rewards is build and analyzed (Lines \ref{alg:line:qstart} to \ref{alg:line:qend}).
$\ma^\ast$ might still contain ECs other than $\set{\state_{\staysign}}$.
Those ECs shall be left eventually to avoid collecting infinite negative reward for a total reward objective $\totfunc{\rewardassignment_i}$. 
Note that the weight $\tupleaccess{\weightvector}{i}$ for such an objective might be zero, \ie the rewards of $\rewardassignment_i$ are not present in $\rewardassignment^\ast$.
It is therefore necessary to explicitly restrict the analysis to strategies  that almost surely (\ie with probability 1) reach $\state_{\staysign}$.
To compute the maximal expected total reward in \Cref{alg:line:valuecomp} with, \eg standard value iteration, we can consider another quotient model for $\ma^*$ and the 0-ECs of $\ma^*$ and $\rewardassignment^\ast$.
In contrast to \Cref{def:ecq}, this quotient should not introduce the $\staysign$ action since it shall not be possible to remain in an EC forever.
In \Cref{alg:line:stratcomp}, the strategies for the 0-ECs and for the quotient $\ma^\ast$ are combined into one strategy $\strat_\weightvector$ for $\ma$.
Here, $\strat_{\component \lozenge \state'}$ refers to a strategy of $\submodel{\ma}{\component}$ for which every state $\state \in \statesof{\component}$ eventually reaches $\state' \in \statesof{\component}$ almost surely.

\begin{algorithm}[t]
	\Input{MA $\ma$ with initial state $\sinit$, objectives $\multiobjquery= \tuple{\totfunc{\rewardassignment_1}, \dots, \totfunc{\rewardassignment_\numtotobj}, \lrafunc{\rewardassignment_{\numtotobj + 1}}, \dots, \lrafunc{\rewardassignment_\numobj} }$, weight vector $\weightvector$ }
	\Output{Value $v_\weightvector$, strategy $\strat_\weightvector$ as in the weighted sum optimization problem}
$\componentset \gets \zeromecs{\ma}{\rewardassignment_1, \dots, \rewardassignment_i}$
\label{alg:line:ecstart}\tcp{Compute maximal 0-ECs and their LRA.}
\ForEach{$\component \in \componentset$}{
	Compute $v_\component = \max \set{ \expval{\submodel{\ma}{\component}}{\strat}(\lrafunc{\rewardassignment_\weightvector^\mathit{lra}}) \mid \strat \in \stratsmd{\submodel{\ma}{\component}}}$ 
	\linebreak
	and inducing strategy $\strat_\component \in \stratsmd{\submodel{\ma}{\component}}$\label{alg:line:ecend}
}
$\ma^\ast \gets \quotient{\ma}{\componentset}$ \label{alg:line:qstart}\tcp{Build and analyze quotient model.}
Build reward assignment $\rewardassignment^\ast$ with 
$\rewardassignment^\ast(\tuple{\state, \actionordur}, \state') =
\begin{cases}
 v_\component & \text{if } \state = \component, \actionordur=\staysign, \text{ and } \state' = \state_\staysign\\
 \rewardassignment_\weightvector^\mathit{tot}(\tuple{\hat{\state}, \action}, \state') & \text{if } \state= \component, \actionordur = \tuple{\hat{\state}, \action} \in \exitsof{\component}\\
 \rewardassignment_\weightvector^\mathit{tot}(\tuple{\state, \action}, \state') & \text{otherwise}
\end{cases}
$\\
Compute $v_\weightvector = \max \set{ \expval{\ma^\ast}{\strat}\!(\totfunc{\rewardassignment^\ast}) \,\middle|\, {\strat \in \stratsmd{\ma^\ast},\ \probmeasure{\ma^\ast}{\strat}\!(\lozenge \set{\state_{\staysign}}) = 1} }$ \linebreak 
and inducing strategy $\strat^\ast \in \stratsmd{\ma^\ast}$
\label{alg:line:valuecomp}\label{alg:line:qend}\\
Build strategy $\strat_\weightvector \in \stratsmd{\ma}$ with \linebreak
$\strat_\weightvector(\state) =
\begin{cases}
\strat_\component(\state) & $\text{if } $\ex{\component \in \componentset} \state \in \statesof{\component} \text{ and } \strat^{\ast}(\component \in \componentset) = \staysign\\
\action & $\text{if } $\ex{\component\in \componentset} \state \in \statesof{\component} \text{ and } \strat^{\ast}(\component) = \tuple{\state,\action}\\
\strat_{\component \lozenge \set{\state'}}(\state) & $\text{if } $\ex{\component\in \componentset} \state \in \statesof{\component} \text{ and } \strat^{\ast}(\component) = \tuple{\state',\action} \text{ for } \state' \neq \state\\
\strat^\ast(\state) & \text{otherwise}
\end{cases}
$\label{alg:line:stratcomp}
\caption{Optimizing the weighted sum for total and LRA objectives}
\label{alg:weighted}
\end{algorithm}

Since \Cref{alg:weighted} produces a memoryless deterministic strategy $\strat_\weightvector$, the point $\point_\weightvector\in\pointdomain$ in \Cref{alg:line:inducedpoint} of \Cref{alg:refinementloop} can be computed on the induced sub-MA for $\strat_\weightvector$.
Assuming exact single-objective solution methods, the resulting value $v_\weightvector$ and strategy $\strat_\weightvector \in \stratsmd{\ma}$ of \Cref{alg:weighted} satisfy $v_\weightvector = \weightvector \cdot \expval{}{\strat_\weightvector}(\multiobjquery)$, yielding an exact solution to the weighted sum optimization problem.
As the number of memoryless deterministic strategies is bounded, we conclude the following,
extending results for pure LRA queries~\cite{BBCFK14} to mixtures with total rewards.
\begin{corollary}
For total- and LRA reward objectives $\multiobjquery$, $\ach{}{\multiobjquery}$ is closed and is the downward convex hull of at most $|\stratsmd{\ma}| = \prod_{\state \in \ps} |\transitions(\state)|$ points.
\end{corollary}

\begin{remark}\label{rem:others}
Our framework can be extended to support objectives beyond total- and LRA rewards.
\emph{Minimizing objectives} where one is interested in a strategy $\strat$ that induces a \emph{small} expected value can be considered by multiplying all rewards with $-1$.
Since we already allow negative values in reward assignments, no further adaptions are necessary.
We emphasize that our framework lifts a restriction imposed in~\cite{FKP12} that disabled a simultaneous analysis of maximizing \emph{and} minimizing total reward objectives.
\emph{Reachability probabilities} can be transformed to expected total rewards on a modified model in which the information whether a goal state has already been visited is stored in the state-space. 
\emph{Goal-bounded} total rewards as in~\cite{GTHRS14}, where no further rewards are collected as soon as one of the goal states is reached can be transformed similarly.
%
 For MDP, \emph{step- and reward-bounded} reachability probabilities can be converted to total reward objectives by unfolding the current amount of steps (or rewards) into the state-space of the model.
Approaches that avoid such an expensive unfolding have been presented in \cite{FKP12} for objectives with step-bounds and in~\cite{HJKQ18,HJKQ20} for objectives with one or multiple reward-bounds. 
\emph{Time-bounded} reachability probabilities for MA have been considered in~\cite{QJK17}. 
%
Finally, $\omega$-regular specifications such as \emph{linear temporal logic (LTL)} formulae have been transformed to total reward objectives in~\cite{FKNPQ11}.
However, the optimization of LRA rewards within the ECs of the model might interfere with the satisfaction of one or more $\omega$-regular specifications~\cite{GZ18}.
\end{remark}

\section{Experimental Evaluation}

\paragraph{Implementation details}
Our approach has been implemented in the model checker \storm~\cite{HJKQV20}.
Given an MA or MDP (specified using the \prism language or \tool{JANI}~\cite{BDHHJT17}), the tool  answers qualitative- and quantitative achievability as well as Pareto queries.
Beside of mixtures of total- and LRA reward objectives, \storm{} also supports most of the extensions in~\Cref{rem:others}---with the notable exception of LTL.
We use LRA value iteration~\cite{BWH17,ACDKM17} and sound value iteration~\cite{QK18} for calls to single-objective model checking. Both provide sound precision guarantees, \ie the relative error of these computations is at most $\varepsilon$, where we set $\varepsilon = 10^{-6}$.

\paragraph{Workstation cluster}
To showcase the capabilities of our implementation, we present a workstation cluster---originally considered in~\cite{HHK00} as a CTMC---now modeled as an MA.
The cluster considers two sub-clusters each consisting of one \emph{switch} and $N$ \emph{workstations}.
Within each sub-cluster the workstations are connected to the switch in a star topology and the two switches are connected with a \emph{backbone}.
Each of the components may fail with a certain rate.
A controller can (i) acquire additional repair units (up to $M$) and (ii) control the movements of the repair units.
In \Cref{fig:cluster} we depict the resulting sets of achievable points---as computed by our implementation---for $N=16$ and $M=4$. 
As objectives, we considered the long-run average number of operating workstations $\lrafunc{\rewardassignment_{\#\mathit{op}}}$, the long-run average probability that at least $N$ workstations are operational $\lrafunc{\rewardassignment_{\#\mathit{op}\ge N}}$, and the total number of acquired repair units $\totfunc{\rewardassignment_{\#\mathit{rep}}}$.

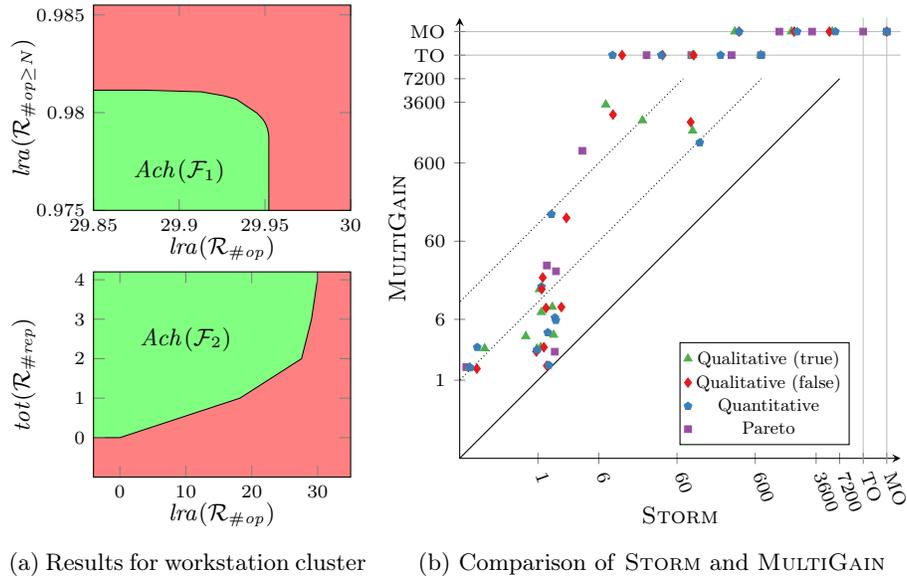
\begin{figure}[t]
	\begin{subfigure}[b]{0.39\linewidth}
		\centering
		\begin{tikzpicture}[scale=1]
			\path[use as bounding box,draw=white] (-1,-0.75) rectangle (3.5,2.75);
			\begin{axis}[
				width=5cm,
				xlabel={$\lrafunc{\rewardassignment_{\#\mathit{op}}}$},
				xmin=29.85,
				xmax = 30,
				xtick={29.85,29.9,29.95,30},
				ylabel={$\lrafunc{\rewardassignment_{\#\mathit{op}\ge N}}$},
				ylabel style={yshift=-0.3cm},
				xlabel style={yshift=0.3cm},
				yticklabel style={font=\scriptsize,/pgf/number format/fixed, /pgf/number format/precision=3},
				xticklabel style={font=\scriptsize},
				ymin=0.975,
				ymax=0.9855,
				ytick={0.975,0.98,0.985},
				axis background/.style={fill=red!50},
				axis on top
				]
				
				\addplot[fill=green!50, very thin] coordinates {
					
					(29.8       , 0.9811460752 )
					(29.88057587, 0.9811460752 )
					(29.9129114, 0.9810632059 )
					(29.92598346, 0.9808597902 )
					(29.93325521, 0.980684152 )
					(29.94538879, 0.9799737714 )
					(29.94891907, 0.9796436213 )
					(29.95006632, 0.9794808775 )				
					(29.95142501, 0.9792010675 )
					(29.95209059, 0.9789202222 )
					(29.95226485, 0.9787210912 )
					(29.95226485, 0.97 )
					(29.8, 0.97 )
				} -- cycle;
				
				\node[] at (axis cs:29.9,0.977) {$\ach{}{\multiobjquery_1}$};
				
			\end{axis} 
		\end{tikzpicture}\\
		\begin{tikzpicture}[scale=1]
			\path[use as bounding box,draw=white] (-1,-0.75) rectangle (3.5,2.75);
			\begin{axis}[
				width=5cm,
				xlabel={$\lrafunc{\rewardassignment_{\#\mathit{op}}}$},
				xmin=-4,
				xmax = 35,
				xtick={0,10,20,30},
				ylabel={$\totfunc{\rewardassignment_{\#\mathit{rep}}}$},
				ylabel style={yshift=-0.3cm},
				xlabel style={yshift=0.3cm},
				yticklabel style={font=\scriptsize,/pgf/number format/fixed, /pgf/number format/precision=3},
				xticklabel style={font=\scriptsize},
				ymin=-1,
				ymax=4.2,
				ytick={0,1,2,3,4},
				axis background/.style={fill=red!50},
				axis on top
				]
				
				\addplot[fill=green!50, very thin] coordinates {
					(29.95226485,           5 )
					(29.95226485,           4 )
					(29.06101055,           3 )
					(27.5501292,           2 )
					(18.20003388,           1 )
					(          0,           0 )
					(          -5,           0 )
					(          -5,           5 )
				} -- cycle;
				
				\node[] at (axis cs:10,2.5) {$\ach{}{\multiobjquery_2}$};
				
			\end{axis} 
		\end{tikzpicture}%
		\caption{Results for workstation cluster}
		\label{fig:cluster}
	\end{subfigure}
	\begin{subfigure}[b]{0.61\linewidth}
		\centering
		\newlength{\scatterplotsize}
\setlength{\scatterplotsize}{\linewidth}
	\begin{tikzpicture}
					\path[use as bounding box,draw=white] (-1,-1) rectangle (6,6);
	\begin{axis}[
		width=\scatterplotsize,
		height=\scatterplotsize,
		axis equal image,
		xmin=0.1,
		ymin=0.1,
		ymax=43200,
		xmax=43200,
		xmode=log,
		ymode=log,
		axis x line=bottom,
		axis y line=left,
		xtick={1,6,60,600,3600,7200},
		xticklabels={1,6,60,600,3600,7200},
		extra x ticks = {14400,28800},
		extra x tick labels = {TO,MO},
		extra x tick style = {grid = major},
		ytick={1,6,60,600,3600,7200},
		yticklabels={1,6,60,600,3600,7200},
		extra y ticks = {14400, 28800},
		extra y tick labels = {TO,MO},
		extra y tick style = {grid = major},
		xlabel={\storm},
		ylabel={\multigain},
		ylabel style={yshift=-0.4cm},
		yticklabel style={font=\scriptsize},
		xticklabel style={rotate=290,anchor=west,font=\scriptsize},
		legend pos=south east,
		legend columns=1,
		legend style={nodes={scale=0.75, transform shape},inner sep=1.5pt,yshift=0cm,xshift=-0.5cm},
		]
		
		\addplot[
		scatter,
		only marks,
		scatter/classes={
			qual1={mark=triangle*, mark size=1.75, color3},
			qual2={mark=diamond*, mark size=1.75, color2},
			quant={mark=pentagon*, mark size=1.50, color1},
			pareto={mark=square*, mark size=1.25, color4}
		},
		scatter src=explicit symbolic
		]%
		table [col sep=semicolon,x=Storm,y=Multigain,meta=scatterclass] {experiments/scatter.csv};
		\legend{Qualitative (true), Qualitative (false), Quantitative, Pareto}
		\addplot[no marks] coordinates {(0.01,0.01) (7200,7200) };
		\addplot[no marks, densely dotted] coordinates {(0.01,0.1) (720,7200)};
		\addplot[no marks, densely dotted] coordinates {(0.001,0.1) (72,7200)};
	\end{axis}
\end{tikzpicture}
		\caption{Comparison of \storm and \multigain}
		\label{fig:scatter}
	\end{subfigure}
	\caption{Exemplary results and runtime comparison with \multigain}
	\label{fig:clusterscatter}
\end{figure}
\begin{table}[t]
	\caption{Results for pure LRA Pareto queries}
	\centering
	\label{tab:paretolra}
	\setlength{\tabcolsep}{4pt}
	\aboverulesep = 0.1pt
	\belowrulesep = 0.1pt
	\scriptsize{
		\begin{tabular}{ccccccrrrrrrrr|r}
\toprule[1.5pt]
\multicolumn{1}{c}{Model} & \multicolumn{1}{c}{Par.} & \multicolumn{1}{c}{\#lra} & \multicolumn{1}{c}{$|\states|$} & \multicolumn{1}{c}{$|\ms|$} & \multicolumn{1}{c}{$|\transitions|$} & \multicolumn{1}{c}{$|\componentset|$} & \multicolumn{1}{c}{$|\states_\componentset|$} & \multicolumn{1}{c}{\#iter} & \multicolumn{1}{c}{\storm} & \multicolumn{1}{c}{\multigain}\\
\midrule[1.5pt]
\benchmark{csn} & 3 & $3$ & $177$ &  & $427$ & $38$ & $158$ & $9$ & \ensuremath{1.23} & \\
\benchmark{csn} & 4 & $4$ & $945$ &  & $2753$ & $176$ & $880$ & $30$ & \ensuremath{109} & \\
\benchmark{csn} & 5 & $5$ & $4833$ &  & $2{\cdot} 10^{4}$ & $782$ & $4622$ &  & TO & \\
\midrule
\benchmark{mut} & 3 & $2$ & $3{\cdot} 10^{4}$ &  & $5{\cdot} 10^{4}$ & $1$ & $3{\cdot} 10^{4}$ & $15$ & \ensuremath{3.7} & \ensuremath{859}\\
\benchmark{mut} & 4 & $2$ & $7{\cdot} 10^{5}$ &  & $1{\cdot} 10^{6}$ & $1$ & $7{\cdot} 10^{5}$ & $14$ & \ensuremath{91.4} & TO\\
\benchmark{mut} & 5 & $2$ & $1{\cdot} 10^{7}$ &  & $3{\cdot} 10^{7}$ & $1$ & $1{\cdot} 10^{7}$ & $12$ & \ensuremath{3197} & MO\\
\midrule
\benchmark{phi} & 4 & $2$ & $9440$ &  & $4{\cdot} 10^{4}$ & $1$ & $9440$ & $6$ & \ensuremath{1.7} & \ensuremath{24.7}\\
\benchmark{phi} & 5 & $2$ & $9{\cdot} 10^{4}$ &  & $4{\cdot} 10^{5}$ & $1$ & $9{\cdot} 10^{4}$ & $18$ & \ensuremath{24.5} & TO\\
\benchmark{phi} & 6 & $2$ & $2{\cdot} 10^{6}$ &  & $1{\cdot} 10^{7}$ & $1$ & $2{\cdot} 10^{6}$ & $12$ & \ensuremath{1221} & MO\\
\midrule
\benchmark{res} & 5-5 & $2$ & $2618$ &  & $8577$ & $1$ & $2618$ & $16$ & \ensuremath{1.64} & \ensuremath{2.31}\\
\benchmark{res} & 15-15 & $2$ & $2{\cdot} 10^{5}$ &  & $7{\cdot} 10^{5}$ & $1$ & $2{\cdot} 10^{5}$ & $3$ & \ensuremath{712} & TO\\
\benchmark{res} & 20-20 & $2$ & $8{\cdot} 10^{5}$ &  & $2{\cdot} 10^{6}$ & $1$ & $8{\cdot} 10^{5}$ & $7$ & \ensuremath{299} & TO\\
\midrule
\benchmark{sen} & 2 & $3$ & $7855$ &  & $2{\cdot} 10^{4}$ & $3996$ & $6105$ & $13$ & \ensuremath{3.41} & \\
\benchmark{sen} & 3 & $3$ & $8{\cdot} 10^{4}$ &  & $3{\cdot} 10^{5}$ & $5{\cdot} 10^{4}$ & $7{\cdot} 10^{4}$ & $14$ & \ensuremath{274} & \\
\benchmark{sen} & 4 & $3$ & $6{\cdot} 10^{5}$ &  & $3{\cdot} 10^{6}$ & $4{\cdot} 10^{5}$ & $5{\cdot} 10^{5}$ &  & TO & \\
\midrule
\benchmark{vir} & 2 & $2$ & $80$ &  & $393$ & $2$ & $66$ & $4$ & \ensuremath{${<}\,$ 1} & \ensuremath{1.47}\\
\benchmark{vir} & 3 & $2$ & $2{\cdot} 10^{4}$ &  & $2{\cdot} 10^{5}$ & $2$ & $2{\cdot} 10^{4}$ & $2$ & \ensuremath{1.3} & \ensuremath{29.3}\\
\benchmark{vir} & 4 & 2 & $4{\cdot} 10^{7}$ &  & $7{\cdot} 10^{8}$ & ? & ? & & MO & MO\\
\midrule[1.5pt]
\benchmark{clu} & 8-3 & $2$ & $2{\cdot} 10^{5}$ & $1{\cdot} 10^{5}$ & $4{\cdot} 10^{5}$ & $4$ & $2{\cdot} 10^{5}$ & $11$ & \ensuremath{287} & \\
\benchmark{clu} & 16-4 & $2$ & $2{\cdot} 10^{6}$ & $9{\cdot} 10^{5}$ & $4{\cdot} 10^{6}$ & $5$ & $2{\cdot} 10^{6}$ & $10$ & \ensuremath{4199} & \\
\benchmark{clu} & 32-3 & $2$ & $2{\cdot} 10^{6}$ & $1{\cdot} 10^{6}$ & $5{\cdot} 10^{6}$ & $4$ & $2{\cdot} 10^{6}$ &  & TO & \\
\midrule
\benchmark{dpm} & 3-3 & $2$ & $2640$ & $1008$ & $3240$ & $1$ & $2640$ & $32$ & \ensuremath{19.5} & \\
\benchmark{dpm} & 4-4 & $2$ & $3{\cdot} 10^{4}$ & $1{\cdot} 10^{4}$ & $4{\cdot} 10^{4}$ & $1$ & $3{\cdot} 10^{4}$ & $33$ & \ensuremath{1179} & \\
\benchmark{dpm} & 5-5 & $2$ & $6{\cdot} 10^{5}$ & $2{\cdot} 10^{5}$ & $7{\cdot} 10^{5}$ & $1$ & $6{\cdot} 10^{5}$ &  & TO & \\
\midrule
\benchmark{pol} & 3-3 & $2$ & $9522$ & $4801$ & $2{\cdot} 10^{4}$ & $1$ & $9522$ & $17$ & \ensuremath{3.44} & \\
\benchmark{pol} & 4-3 & $2$ & $5{\cdot} 10^{4}$ & $3{\cdot} 10^{4}$ & $1{\cdot} 10^{5}$ & $1$ & $5{\cdot} 10^{4}$ & $19$ & \ensuremath{19.2} & \\
\benchmark{pol} & 4-4 & $2$ & $8{\cdot} 10^{5}$ & $5{\cdot} 10^{5}$ & $2{\cdot} 10^{6}$ & $1$ & $8{\cdot} 10^{5}$ & $29$ & \ensuremath{3350} & \\
\midrule
\benchmark{rqs} & 2-2 & $2$ & $1619$ & $628$ & $2296$ & $1$ & $1618$ & $63$ & \ensuremath{4.52} & \\
\benchmark{rqs} & 3-3 & $2$ & $9{\cdot} 10^{4}$ & $4{\cdot} 10^{4}$ & $1{\cdot} 10^{5}$ & $1$ & $9{\cdot} 10^{4}$ & $106$ & \ensuremath{162} & \\
\benchmark{rqs} & 5-3 & $2$ & $2{\cdot} 10^{6}$ & $1{\cdot} 10^{6}$ & $4{\cdot} 10^{6}$ & $1$ & $2{\cdot} 10^{6}$ & $97$ & \ensuremath{4345} & \\
\bottomrule[1.5pt]
\end{tabular}

	}	
\end{table}
\begin{table}[t]
	\caption{Results for Pareto queries with other objective types}
	\centering
	\label{tab:paretoother}
	\setlength{\tabcolsep}{4pt}
	\aboverulesep = 0.1pt
\belowrulesep = 0.1pt
	\scriptsize{
		\begin{tabular}{ccccccrrrrrrrr}
\toprule[1.5pt]
\multicolumn{1}{c}{Model} & \multicolumn{1}{c}{Par.} & \multicolumn{1}{c}{\#lra/tot/bnd} & \multicolumn{1}{c}{$|\states|$} & \multicolumn{1}{c}{$|\ms|$} & \multicolumn{1}{c}{$|\transitions|$} & \multicolumn{1}{c}{$|\componentset|$} & \multicolumn{1}{c}{$|\states_\componentset|$} & \multicolumn{1}{c}{\#iter} & \multicolumn{1}{c}{\storm}\\
\midrule[1.5pt]
\benchmark{res} & 5-5 & 2-0-1 & $2618$ &  & $8577$ & $1$ & $2618$ & $17$ & \ensuremath{4.27}\\
\benchmark{res} & 5-5 & 2-1-0 & $2618$ &  & $8577$ & $1$ & $1705$ & $6$ & \ensuremath{1.43}\\
\benchmark{res} & 15-15 & 2-0-1 & $2{\cdot} 10^{5}$ &  & $7{\cdot} 10^{5}$ & $1$ & $2{\cdot} 10^{5}$ & $4$ & \ensuremath{792}\\
\benchmark{res} & 15-15 & 2-1-0 & $2{\cdot} 10^{5}$ &  & $7{\cdot} 10^{5}$ & $1$ & $1{\cdot} 10^{5}$ & $8$ & \ensuremath{1061}\\
\benchmark{res} & 20-20 & 2-0-1 & $8{\cdot} 10^{5}$ &  & $2{\cdot} 10^{6}$ & $1$ & $8{\cdot} 10^{5}$ & $8$ & \ensuremath{641}\\
\benchmark{res} & 20-20 & 2-1-0 & $8{\cdot} 10^{5}$ &  & $2{\cdot} 10^{6}$ & $1$ & $4{\cdot} 10^{5}$ & $4$ & \ensuremath{101}\\
\midrule[1.5pt]
\benchmark{clu} & 8-3 & 1-1-0 & $2{\cdot} 10^{5}$ & $1{\cdot} 10^{5}$ & $4{\cdot} 10^{5}$ & $4$ & $2{\cdot} 10^{5}$ & $7$ & \ensuremath{163}\\
\benchmark{clu} & 16-4 & 1-1-0 & $2{\cdot} 10^{6}$ & $9{\cdot} 10^{5}$ & $4{\cdot} 10^{6}$ & $5$ & $2{\cdot} 10^{6}$ & $9$ & \ensuremath{3432}\\
\benchmark{clu} & 32-3 & 1-1-0 & $2{\cdot} 10^{6}$ & $1{\cdot} 10^{6}$ & $5{\cdot} 10^{6}$ & $4$ & $2{\cdot} 10^{6}$ & $7$ & \ensuremath{3328}\\
\midrule
\benchmark{dpm} & 3-3 & 1-0-1 & $5232$ & $1980$ & $6408$ & $46$ & $3045$ & $2$ & \ensuremath{11.2}\\
\benchmark{dpm} & 3-3 & 1-1-0 & $4584$ & $1656$ & $5562$ & $25$ & $2856$ & $4$ & \ensuremath{${<}\,$ 1}\\
\benchmark{dpm} & 4-4 & 1-0-1 & $7{\cdot} 10^{4}$ & $2{\cdot} 10^{4}$ & $8{\cdot} 10^{4}$ & $497$ & $4{\cdot} 10^{4}$ & $2$ & \ensuremath{214}\\
\benchmark{dpm} & 4-4 & 1-1-0 & $6{\cdot} 10^{4}$ & $2{\cdot} 10^{4}$ & $7{\cdot} 10^{4}$ & $301$ & $4{\cdot} 10^{4}$ & $4$ & \ensuremath{3.32}\\
\benchmark{dpm} & 5-5 & 1-0-1 & $1{\cdot} 10^{6}$ & $3{\cdot} 10^{5}$ & $1{\cdot} 10^{6}$ & $6476$ & $6{\cdot} 10^{5}$ &  & TO\\
\benchmark{dpm} & 5-5 & 1-1-0 & $1{\cdot} 10^{6}$ & $3{\cdot} 10^{5}$ & $1{\cdot} 10^{6}$ & $4321$ & $6{\cdot} 10^{5}$ & $4$ & \ensuremath{329}\\
\midrule
\benchmark{pol} & 3-3 & 1-1-0 & $1{\cdot} 10^{4}$ & $5309$ & $2{\cdot} 10^{4}$ & $1$ & $9522$ & $3$ & \ensuremath{1.37}\\
\benchmark{pol} & 4-3 & 1-1-0 & $6{\cdot} 10^{4}$ & $3{\cdot} 10^{4}$ & $1{\cdot} 10^{5}$ & $1$ & $5{\cdot} 10^{4}$ & $3$ & \ensuremath{2.52}\\
\benchmark{pol} & 4-4 & 1-1-0 & $9{\cdot} 10^{5}$ & $5{\cdot} 10^{5}$ & $2{\cdot} 10^{6}$ & $1$ & $8{\cdot} 10^{5}$ & $3$ & \ensuremath{237}\\
\midrule
\benchmark{rqs} & 2-2 & 1-1-0 & $2805$ & $1039$ & $4159$ & $1$ & $1618$ & $3$ & \ensuremath{${<}\,$ 1}\\
\benchmark{rqs} & 3-3 & 1-1-0 & $1{\cdot} 10^{5}$ & $6{\cdot} 10^{4}$ & $3{\cdot} 10^{5}$ & $1$ & $9{\cdot} 10^{4}$ & $3$ & \ensuremath{4.51}\\
\benchmark{rqs} & 5-3 & 1-1-0 & $3{\cdot} 10^{6}$ & $2{\cdot} 10^{6}$ & $7{\cdot} 10^{6}$ & $1$ & $2{\cdot} 10^{6}$ & $3$ & \ensuremath{182}\\
\bottomrule[1.5pt]
\end{tabular}

	}
\end{table}

\paragraph{Related tools}
\multigain~\cite{BCFK15} is an extension of \prism~\cite{KNP11} that implements the LP-based approach of~\cite{BBCFK14} for multiple LRA objectives on MDP to answer qualitative and quantitative achievability as well as Pareto queries.
For the latter, it is briefly mentioned in~\cite{BCFK15} that ideas of~\cite{FKP12} were used similar to our approach but no further details are provided.
\multigain does not support MA, \emph{mixtures} with total reward objectives,
and Pareto queries with $\numobj > 2$ objectives.
However, it does support more general quantitative achievability queries.

\prismgames~\cite{KPW18,KNPS20} implements value iteration over convex sets~\cite{BKTW15,BKW18} to analyze multiple LRA reward objectives on stochastic games (SGs).
By converting MDPs to 1-player SGs, \prismgames could also be applied in our setting.
However, some experiments on 1-player SGs indicated that this approach is not competitive compared to the dedicated MDP implementations in \multigain and \storm.
We therefore do not consider \prismgames in our evaluation.

\paragraph{Benchmarks}
We consider 10 different case studies including the workstation cluster~(\benchmark{clu}) as well as benchmarks from QVBS~\cite{HKPQR19} (\benchmark{dpm}, \benchmark{rqs}, \benchmark{res}), from \multigain~\cite{BCFK15} (\benchmark{mut}, \benchmark{phi}, \benchmark{vir}), from~\cite{KM17} (\benchmark{csn}, \benchmark{sen}), and from~\cite{QJK17} (\benchmark{pol}).
For each case study we consider 3 concrete instances resulting in 12 MAs and 18 MDPs.
The analyzed objectives range over LRA rewards, (goal-bounded) total rewards, and time-, step- and unbounded reachability probabilities.

\paragraph{Set-up}
We evaluated the performance of \storm and \multigain Version 1.0.2\footnote{Obtained from \url{http://qav.cs.ox.ac.uk/multigain} and invoked with  \tool{Gurobi}~\cite{gurobi}.}.
All experiments were run on 4 cores\footnote{\storm uses one core, \multigain uses multiple cores due to \tool{Java}'s garbage collection and \tool{Gurobi}'s parallel solving techniques.} of an Intel Xeon Platinum 8160 CPU with a time limit of 2 hours and 32 GB RAM.
For each experiment we measured the total runtime (including model building) to solve one query.
For qualitative and quantitative achievability we consider thresholds close to the Pareto front.
For Pareto queries, the approximation precision $10^{-4}$ was set to both tools.

\paragraph{Results}
\Cref{fig:scatter} visualizes the runtime comparison with \multigain. A point $\tuple{x,y}$ in the plot corresponds to a query that has been solved by $\storm$ in $x$ seconds and by $\multigain$ in $y$ seconds.
Points on the solid diagonal mean that both tools were equally fast.
The two dotted lines indicate experiments where \storm only required $\frac{1}{10}$ resp. $\frac{1}{100}$ of the time of \multigain.
TO and MO indicate a time- or memory out.
\Cref{tab:paretolra,tab:paretoother} provide further data for Pareto queries.
The columns indicate model name and parameters, the number of LRA reward, total reward, and bounded reachability objectives, the number of states ($|\states|$), Markovian states ($|\ms|$), successor distributions ($|\transitions|$), 0-ECs ($|\componentset|$), and states within 0-ECs ($|\states_\componentset|$) of the MA or MDP, the number of iterations (\#iters) of~\Cref{alg:refinementloop} performed by \storm, and the total runtime of \storm and \multigain in seconds.
Runtimes are omitted if the tool does not support the query.
MDP (MA) benchmarks are at the top (bottom) of each table.
\Cref{tab:paretolra} considers pure LRA queries, whereas \Cref{tab:paretoother} considers mixtures.

\paragraph{Discussion}
As indicated in~\Cref{fig:scatter}, our implementation outperforms \multigain on almost all benchmarks and for all types of queries and is often orders of magnitude faster.
According to \multigain's log files, the majority of its runtime is spend for solving LPs, suggesting that the better performance of \storm is likely due to the iterative approach presented in this work.

\Cref{tab:paretolra} shows that \emph{pure LRA queries on models with millions of states can be handled}.
There were no significant runtime gaps between MA and MDP models.
For \benchmark{csn}, the increased number of objectives drastically increases the overall runtime.
This is partly due to our naive implementation of the geometric set representations used in \Cref{alg:refinementloop}.
\Cref{tab:paretoother} indicates that the performance and scalability for mixtures of LRA and other types of objectives is similar.
One exception are queries involving time-bounded reachability on MA (\eg \benchmark{dpm}). 
Here, our implementation is based on the single-objective approach of~\cite{GHHKT14} that is known to be slower than more recent methods~\cite{BHHK15,BF19}.

\paragraph{Data availability} The implementation, models, and log files are available at~\cite{artifact}.
\section{Conclusion}
The analysis of multi-objective model checking queries involving multiple long-run average rewards can be incorporated into the framework of~\cite{FKP12} enabling (i) the use of off-the-shelf single-objective algorithms for LRA and (ii) the combination with other kinds of objectives such as total rewards.
Our experiments indicate that this approach clearly outperforms existing algorithms based on linear programming.
Future work includes lifting the approach to \emph{partially observable MDP} and \emph{stochastic games}, potentially using ideas of \cite{BJKQ20} and \cite{ACKWW20}, respectively.

\iftoggle{TR}{}{\newpage}

\bibliographystyle{splncs04}
\bibliography{main}

\iftoggle{TR}{}{

\vfill

{\small\medskip\noindent{\bf Open Access} This chapter is licensed under the terms of the Creative Commons\break Attribution 4.0 International License (\url{https://creativecommons.org/licenses/by/4.0/}), which permits use, sharing, adaptation, distribution and reproduction in any medium or format, as long as you give appropriate credit to the original author(s) and the source, provide a link to the Creative Commons license and indicate if changes were made.}

{\small \spaceskip .28em plus .1em minus .1em The images or other third party material in this chapter are included in the chapter's Creative Commons license, unless indicated otherwise in a credit line to the material.~If material is not included in the chapter's Creative Commons license and your intended\break use is not permitted by statutory regulation or exceeds the permitted use, you will need to obtain permission directly from the copyright holder.}

\medskip\noindent\includegraphics{cc_by_4-0.eps}
}

\end{document}